\documentclass[11pt]{article}
\usepackage[margin = 1in]{geometry}

\usepackage{amsmath,amssymb,amsthm,hyperref,color}
\usepackage{float}
\usepackage[font=small]{caption}
\usepackage[font=footnotesize]{subcaption}
\usepackage[inline]{enumitem}
\usepackage{filecontents}
\usepackage{tikz}
\hypersetup{colorlinks=true,citecolor=blue, linkcolor=blue,
urlcolor=blue}

\newtheorem{theorem}{Theorem}
\newtheorem{lemma}[theorem]{Lemma}

\newtheorem{corollary}[theorem]{Corollary}

\newtheorem{definition}[theorem]{Definition}
\newtheorem{example}[theorem]{Example}

\newtheorem*{lemvertmatch}{Lemma~\ref{lem:vertmatch}}
\newtheorem*{lemexpvertmatch}{Lemma~\ref{lem:expvertmatch}}
\newtheorem*{thmhard}{Theorem~\ref{thm:hard}}
\newtheorem*{thmhardsign}{Theorem~\ref{thm:hardsign}}
\newtheorem*{thmconnconst}{Theorem~\ref{thm:connconst}}

\renewcommand\arg{\text{arg}}

\def\numP{\#\mathsf{P}}

\def\Reals{\mathbb{R}}

\def\Complex{\mathbb{C}}

\newcommand{\size}[1]{\mathrm{size}(#1)}

\newcommand{\im}{\mathrm{i}}

\newcommand{\depth}{\mathrm{depth}}

\def\Matchings#1{\#\ensuremath{\mathsf{NormMatchings}(\gamma,\Delta,#1)}}
\def\SignMatchings{\ensuremath{\mathsf{SignMatchings}(\gamma,\Delta)}}

\newcommand{\eps}{\epsilon}

\makeatletter
\def\prob#1#2#3{\goodbreak\begin{list}{}{\labelwidth\z@ \itemindent-\leftmargin
                        \itemsep\z@  \topsep6\p@\@plus6\p@
                        \let\makelabel\descriptionlabel}
                \item[\it Name]#1
               \item[\it Instance]                #2
                \item[\it Output]#3
                \end{list}}
\makeatother

\def\Bezakova{Bez\'{a}kov\'{a}}

\title{The complexity of approximating the matching polynomial in the complex plane\thanks{A preliminary short version of these results  appeared in the proceedings of \emph{ICALP 2019} (Track A).}}
\author{
Ivona \Bezakova\thanks{
  Department of Computer Science, Rochester Institute of Technology, Rochester, NY, USA. Research
supported by NSF grant CCF-1319987.} \and
Andreas Galanis\thanks{
  Department of Computer Science, University of Oxford, Wolfson Building, Parks Road, Oxford, OX1~3QD, UK.
  The research leading to these results has received funding from the European Research Council under
  the European Union's Seventh Framework Programme (FP7/2007-2013) ERC grant agreement no.\ 334828. The paper
  reflects only the authors' views and not the views of the ERC or the European Commission.
  The European Union is not liable for any use that may be made of the information contained therein.}
  \and
  Leslie Ann Goldberg$^\ddag$
\and
 Daniel \v{S}tefankovi\v{c}\thanks{
Department of Computer Science, University of Rochester,
Rochester, NY 14627.  Research
supported by NSF grant CCF-1563757.}
 }

\date{11 January 2021}

\begin{document}

\maketitle
\thispagestyle{empty}
\begin{abstract} 
We study the problem of approximating the value of the matching polynomial on graphs with edge parameter $\gamma$, where $\gamma$ takes arbitrary values in the complex plane. 

When $\gamma$ is a positive real, Jerrum and Sinclair showed that the problem admits an FPRAS on general graphs. For general complex values of $\gamma$, Patel and Regts, building on methods developed by Barvinok, showed that the problem admits an FPTAS on graphs of maximum degree $\Delta$ as long as $\gamma$ is not a negative real number less than or equal to $-1/(4(\Delta-1))$. Our first main result completes the picture for the approximability of the matching polynomial on bounded degree graphs. We show that for all $\Delta\geq 3$ and all real $\gamma$ less than $-1/(4(\Delta-1))$, the problem of approximating the value of the matching polynomial on graphs of maximum degree $\Delta$ with edge parameter $\gamma$ is \#P-hard. 

We then explore whether the maximum degree parameter can be replaced by the connective constant. Sinclair et al. showed that for positive real $\gamma$ it is possible to approximate the value of the matching polynomial using a correlation decay algorithm on graphs with bounded connective constant (and potentially unbounded maximum degree). We first show that this result does not extend in general in the complex plane; in particular, the problem is \#P-hard on graphs with bounded connective constant for a dense set of $\gamma$ values on the negative real axis. Nevertheless, we show that the result does extend for any complex value $\gamma$ that does not lie on the negative real axis. Our analysis accounts for complex values of $\gamma$ using geodesic distances in the complex plane in the metric defined by an appropriate density function.
\end{abstract}

\newpage

\clearpage
 \setcounter{page}{1}
\section{Introduction}

We study the problem of approximating the matching polynomial of a graph. This polynomial
has a parameter~$\gamma$, called the edge activity.
A \emph{matching} of a graph~$G$ is a set $M\subseteq E(G)$ 
such that each vertex $v\in V(G)$ is contained in at most one edge in~$M$.
We denote by $\mathcal{M}_G$ the set of all matchings of~$G$.
The matching polynomial $Z_G(\gamma)$ is given by 
\[Z_G(\gamma)=\sum_{M\in \mathcal{M}_G}\gamma^{|M|}.\]
This polynomial is also referred to as the partition function of the \emph{monomer-dimer model}
in statistical physics.

Here is what is known about approximating this polynomial.
We first describe the case where $\gamma$ is positive and real.
This is a natural case, and is the case where the first complexity-theoretic results were obtained.
We next describe the more general case, where $\gamma$ is a complex number.
There are many reasons for considering the more general case.
The parameter $\gamma$ is defined to be complex, rather than real, in the classic paper of Heilmann and Lieb~\cite{Heilmann1972}.
Furthermore, it has recently been shown~\cite{physics} that ``the quantum evolution of a system originally in thermodynamic equilibrium
is equivalent to the partition function of the system with a complex parameter''. 
As~\cite{physics} explains, recent discoveries in physics ``make it possible to
study thermodynamics in the complex plane of physical parameters'' --- so complex parameters are increasingly relevant.
As we will see in this paper, 
it is beneficial to study partition functions with complex parameters even when one is most interested in the real case --- 
the reason is that the generalisation sheds light on ``what is really going on'' with complexity bottlenecks, and on
appropriate potential functions. Here is the summary of known results in both cases.
 
\begin{itemize}
\item {\bf When the edge activity~$\gamma$ is a positive real number:\quad}
For any positive real number~$\gamma$,
Jerrum and Sinclair~\cite[Corollary 4.4]{JSPermanent} gave an FPRAS for approximating 
$Z_G(\gamma)$.
Using the correlation decay technique,
Bayati et al.~\cite{Bayati} gave a (deterministic) FPTAS for the same problem
for the case in which the degree of the input graph~$G$ is  
at most a
constant~$\Delta$. 
\item {\bf When the edge activity~$\gamma$ is a complex number:\quad}
Known results are restricted to the case where $\gamma$ is
 not a real number less than or equal to $-1/(4(\Delta-1))$.
 In this case, there is a positive result, due to Patel and Regts~\cite{PR}.
 Using a method of  Barvinok~\cite{BarvinokPermanent, barvinokbook}
for approximating a partition function by truncating its Taylor series
(in a region where the partition function has no zeroes),
Patel and Regts \cite[Theorem 1.2]{PR}  extended the positive result of Bayati et al.  to 
the case in which $\gamma$ is a complex number that is not a negative real that is less
than $-1/(4(\Delta-1))$,  see also \cite[Section 5.1.7]{barvinokbook}). Patel and Regts obtained
a polynomial time algorithm (rather than a quasi-polynomial time one)
by developing clever methods for exactly computing coefficients
of the Taylor series.
 \end{itemize}  
   
 Our first contribution completes  this picture by showing 
 that for all $\Delta \geq 3$ and all real $\gamma < -1/(4(\Delta-1))$ 
 it is actually \#P-hard to approximate $Z_G(\gamma)$ on graphs with degree at most~$\Delta$. 
  We  use the  following notation to 
  state our result more precisely.  
  We consider the problems of multiplicatively approximating the norm of $Z_G(\gamma)$, and 
  of computing its sign.   Our first theorem shows that,
  for all $\Delta \geq 3$ and all rational numbers $\gamma < -1/(4(\Delta-1))$,
  it is $\numP$-hard to approximate $|Z_G(\gamma)|$ on bipartite graphs of maximum degree $\Delta$ within a constant factor.
\begin{theorem}\label{thm:hard}
Let $\Delta\geq 3$ and $\gamma<-\frac{1}{4(\Delta-1)}$ be a rational number. Then, it is $\numP$-hard to approximate $|Z_G(\gamma)|$ within a factor of~1.01 on graphs $G$ of maximum degree $\Delta$, even when restricted to bipartite graphs $G$ with $Z_G(\gamma)\neq 0$.
\end{theorem}

The number~1.01 in Theorem~\ref{thm:hard} is not important.  
It can be replaced with any constant greater than~$1$. In fact,
for any fixed $\epsilon>0$, the theorem, together with a standard powering argument,
 shows that it is \#P-hard to approximate~$Z_G(\gamma)$ within a factor of $2^{|V(G)|^{1-\epsilon}}$.

Our second theorem shows that it is $\numP$-hard to compute the sign of  $Z_{G}(\gamma)$ on bipartite graphs of maximum degree $\Delta$. 
\begin{theorem}\label{thm:hardsign}
Let $\Delta\geq 3$ and $\gamma<-\frac{1}{4(\Delta-1)}$ be a rational number. Then,  it is $\numP$-hard to decide whether $Z_G(\gamma)>0$  on graphs $G$ of maximum degree $\Delta$, even when restricted to bipartite graphs $G$ with $Z_G(\gamma)\neq 0$.
\end{theorem}

We next explore whether the  bound on the maximum degree of~$G$ can be relaxed to a restriction on average degree.
The notion of average degree that we use is the \emph{connective constant}. Given a graph~$G$, and a  vertex $v$, 
let   $N_G(v,k)$ be the number of $k$-edge paths in $G$   that start from $v$.
The following definition is 
  taken almost verbatim  from \cite{connconst2,connconst1}.\footnote{The only difference 
between Definition~\ref{def:cc} and  the corresponding definitions in \cite{connconst2,connconst1} is the addition of the terminology ``profile $(a,c)$'' which will be used to state our hardness results in a strong form (the results in \cite{connconst2,connconst1} were algorithmic which is why this handle on the constants $a$ and $c$ was not required). 
} 
\begin{definition}[{\cite{connconst2,connconst1}}]\label{def:cc}
Let $\mathcal{F}$ be a family of finite graphs and let $\Delta$, $a$ and $c$ be positive real numbers. The connective constant of $\mathcal{F}$ is at most $\Delta$ with profile $(a,c)$  if, for any graph $G=(V,E)$ in $\mathcal{F}$ and any vertex $v$ in $G$, it holds that $\sum^\ell_{k=1} N_G(v,k)\leq c \Delta^\ell$ for all $\ell \geq a\log |V|$. 
\end{definition}

Sinclair, Srivastava, \v{S}tefankovi\v{c} and Yin~\cite[Theorem 1.3]{connconst2} showed
that, for fixed~$\Delta$, when $\gamma$ is a positive real,
the correlation decay method gives an FPTAS for approximating $Z_G(\gamma)$ on graphs~$G$ with connective constant at
most~$\Delta$ (without any bound on the maximum degree of~$G$).
The run-time of their algorithm is ${(n/\epsilon)}^{{\rm O}(\sqrt{\gamma \Delta} \log \Delta)}$, where $n$ is the number of vertices of $G$ and $\epsilon$ is the relative error.

Our next result shows that, in striking contrast to the bounded-degree case,
the algorithmic result of Sinclair et al.\ cannot be extended to negative reals, even if $\gamma \geq -1/(4(\Delta-1))$.
Given positive real numbers $a$ and $c$ and a real number
$\Delta>1$, let $\mathcal{F}_{\Delta, a,c}$ be
the set of graphs with connective constant at most $\Delta$ and profile $(a,c)$.
 
\begin{theorem}\label{thm:hardconnconst}
There exist a dense set of values $\gamma$ on the negative real axis such that the following holds for any real numbers $\Delta>1$ and all $a,c>0$.
 
\begin{enumerate}
\item\label{it:14g442223} It is \#P-hard to approximate $|Z_G(\gamma)|$ within a factor 1.01 on graphs $G\in \mathcal{F}_{\Delta, a,c}$,
\item\label{it:14g442223b}  it is \#P-hard to decide whether $Z_G(\gamma)>0$ on graphs $G\in \mathcal{F}_{\Delta, a,c}$.
\end{enumerate}
Both of these results hold even when restricted to bipartite graphs $G$ with $Z_G(\gamma)\neq 0$.
\end{theorem}

The algorithmic contribution of our paper is to show that, despite the hardness result of Theorem~\ref{thm:hardconnconst},
correlation decay gives a good approximation algorithm for any complex value $\gamma$ that does not lie on the negative real axis when the input graph has bounded connective constant. 
 It is interesting that we are able to use correlation decay to get a good approximation
 for all non-real complex values~$\gamma$.
 Our result is the only known approximation in this setting. In particular,
it is not known how to obtain such a result using the method of Patel and Regts~\cite{PR}.
In order to describe our result, we use the following notation.
Given a complex number~$x$, let
 $\mathrm{arg}(x)$  denote the principal value of its argument in the range $[0,2\pi)$ and $|x|$ denote its norm. 
Our result is the following.

\newcommand{\statethmconnconst}{
Let $\Delta$, $a$ and $c$ be positive real numbers and let 
$\gamma\in\Complex\setminus  \Reals_{< 0}$ be any fixed edge activity.
Then  
there is
an algorithm
which takes as input an $n$-vertex  graph $G\in \mathcal{F}_{\Delta,a,c}$  and 
a rational $\epsilon\in(0,1)$ and
produces an output $\hat{Z} = Z_G(\gamma) {\mathrm{e}}^z$
for some complex number~$z$ with  $|z| \leq \epsilon$. 
The running time of the algorithm  is  $(\hat{c} n/\epsilon)^{O\big((1+a+\sqrt{|\hat{\gamma}| \Delta}) \log \Delta\big)}$ where $\hat{\gamma}=\frac{2|\gamma|}{1+\cos(\mathrm{arg}\gamma)}$ and $\hat{c}=\max\{1,c\}$.
}
\begin{theorem}\label{thm:connconst}
\statethmconnconst
\end{theorem}

Theorem~\ref{thm:connconst} gives an algorithmic result  which contrasts
with the hardness results of Theorems~\ref{thm:hard} and~\ref{thm:hardsign}.
  It has the following corollary.

\begin{corollary}\label{cor:connconst}
Let $\Delta$, $a$ and $c$ be positive real numbers and let 
$\gamma\in\Complex\setminus  \Reals_{< 0}$ be any fixed edge activity.
Then, for any rational $K>1$  and any positive rational $\rho$, 
there are polynomial-time algorithms
that take as input a graph $G\in \mathcal{F}_{\Delta,a,c}$ 
and approximate
 $|Z_G(\gamma)|$ within a multiplicative factor of $K$ and $\mathrm{arg}(Z_G(\gamma))$ within 
 an additive error $\rho$. 
\end{corollary}

In order to prove Theorem~\ref{thm:connconst}, showing correlation decay for complex~$\gamma$, we use geodesic distances in the complex plane in the metric defined by an appropriate density function. Correlation decay for complex activities has been analysed in the context of the hard-core model (see Harvey, Srivastava and Vondr\'ak \cite{Piyush})\footnote{Note that Harvey et al were actually working with the mutivariate
hard-core polynomial -- this causes interesting complications
which will not be relevant for this paper. They also extend their method (for the hard-core polynomial, in their region) to graphs of
unbounded degree that have bounded connective constant.}. The region in the complex plane in which the authors of \cite{Piyush} worked allowed  them to measure distances using the norm instead of requiring geodesic distances. An alternative approach was given by Peters and Regts \cite{Peters}, again in the context of the hard-core model, where they showed contraction within the basin of an attracting fixpoint using the theory of complex dynamical systems.

\section{Preliminaries}\label{sec:prelims}

Let $\gamma$ be a complex number  and $G=(V,E)$ be an arbitrary graph. Recall that $\mathcal{M}_G$ is the set of matchings of $G$. For a matching $M\in \mathcal{M}_G$, we denote by $\mathsf{ver}(M)$  the set of matched vertices in the matching $M$.  For a vertex $u$ in $G$, we also define
\begin{equation*}
Z_{G,u}(\gamma):=\sum_{M\in \mathcal{M}_G;\, u\in \mathsf{ver}(M)}\gamma^{|M|}\quad \mbox{and} \quad Z_{G,\neg u}(\gamma):=\sum_{M\in \mathcal{M}_G;\, u\notin \mathsf{ver}(M)}\gamma^{|M|}.
\end{equation*}
Thus, $Z_{G,u}(\gamma)$ is the contribution to the partition function $Z_G(\gamma)$ from those matchings $M\in \mathcal{M}_G$ such that  $u$ is matched in $M$, while $Z_{G,\neg u}(\gamma)$ is the contribution to the partition function $Z_G(\gamma)$ from those matchings $M\in \mathcal{M}_G$ such that $u$ is not matched in $M$. 

We will use the following result about the location of the zeroes of the matching polynomial.
\begin{theorem}[{\cite{Heilmann1972}}, see, e.g., {\cite[Theorem 5.1.2]{barvinokbook}}]\label{thm:nozeros}
Let $\Delta\geq 3$ be an integer and  $G$ be a graph of maximum degree $\Delta$. Then, for all complex $\gamma$ that do not lie in the interval $(-\infty,-\frac{1}{4(\Delta-1)})$ of the negative real axis, it holds that $Z_G(\gamma)\neq 0$.
\end{theorem} 

\begin{corollary}\label{lem:positive}
Let $\Delta\geq 3$ be an integer and $\gamma>-\frac{1}{4(\Delta-1)}$ be a real number. Then, for all graphs $G$ of maximum degree $\Delta$ it holds that $Z_{G}(\gamma)>0$.
\end{corollary}

For our approximation algorithm of Theorem~\ref{thm:connconst}, given a graph $G=(V,E)$ with $Z_G(\gamma)\neq 0$ and a vertex $v\in V$, we will be interested in the quantity
\[p_v(G,\gamma):=Z_{G,\neg v}(\gamma)/Z_G(\gamma).\]
The algorithm will be based on the following result by Godsil.
\begin{theorem}[\cite{godsil81}]\label{thm:selfavoid}
Let $\gamma\in \mathbb{C}\backslash \mathbb{R}_{<0}$. Let $G=(V,E)$ be a graph and let $v\in V$ be one of its vertices. Let $T_{SAW}(v,G)$ be the self-avoiding walk tree of $G$ rooted at $v$.  Then,
\[p_v(G,\gamma) = p_v(T_{SAW}(v,G),\gamma).\]
\end{theorem}

\section{FPTAS for graphs with bounded connective constant}

In this section, we prove Theorem~\ref{thm:connconst}. Consider $\gamma\in \mathbb{C}\backslash \mathbb{R}_{<0}$.  

We will use the correlation decay technique of Weitz \cite{Weitz}, which we adapt for use with complex activities. We review the basic idea behind the technique (see, e.g., \cite{Bayati,connconst1,connconst2}). For a graph $G$ (of bounded connective constant), we first express $Z_{G}(\gamma)$ as a telescoping product 
\begin{equation}\label{eq:4bt45tb4tr899}
Z_{G}(\gamma)=1/\prod^n_{i=1}p_{v_j}(G_j,\gamma)
\end{equation}
where $v_1,\hdots,v_n$ is an arbitrary enumeration of the vertices of the graph $G$ and $G_j$ is the graph obtained from $G$ by deleting the vertices $v_1,\hdots,v_j$. In light of  \eqref{eq:4bt45tb4tr899}, we can therefore focus on approximating the value $p_v(G,\gamma)$ for a graph $G$ and vertex $v$. Using Godsil's Theorem (cf. Theorem~\ref{thm:selfavoid}), it in turn suffices to approximate $p_v(T_{SAW}(v,G),\gamma)$. This  might seem as a somewhat simpler task given that $T_{SAW}(v,G)$ is a tree; the caveat however is that the tree $T_{SAW}(v,G)$ is  prohibitively large, so in order to be able to perform computations efficiently we need to truncate  the tree. The correlation decay technique analyses the approximation error introduced by this truncation process by recursively tracking the error using tree recurrences. 

In the case of matchings, for a tree $T$ and a vertex $v$ in $T$, we can write a recursion for $p_v(T,\gamma)$ as follows. If $v$ is the only vertex in $T$, then $p_v(T,\gamma)=1$ (since the only possible matching is the empty set and thus $Z_{T,\neg v}(\gamma) = Z_T(\gamma)=1$). Otherwise, let $T_1,\hdots, T_d$ be the trees of $T\backslash \{v\}$ and let $v_1,\hdots,v_d$ be the neighbours of $v$ in $T_1,\hdots, T_d$, respectively. Then, we have that
\[Z_{T,\neg v}(\gamma)=\prod_{i=1}^d Z_{T_i}(\gamma), \quad Z_{T}(\gamma)=\prod_{i=1}^d Z_{T_i}(\gamma) + \sum_{i=1}^d \gamma\, Z_{T_i,\neg v_i}(\gamma)\prod_{j\in\{1,\dots,d\},j\neq i}Z_{T_j}(\gamma)\]
and therefore
\begin{equation*}
p_v(T,\gamma) = \frac{Z_{T,\neg v}(\gamma)}{Z_T(\gamma)} = \frac{1}{1+\gamma\sum_{i=1}^d \frac{Z_{T_i,\neg v_i}(\gamma)}{Z_{T_i}(\gamma)}} = \frac{1}{1+\gamma\sum_{i=1}^d p_{v_i}(T_i,\gamma)}.
\end{equation*}
Hence, we need to evaluate the recurrence
\begin{equation}\label{eq:xrecurrence}
x = F(x_1,\hdots, x_d) \mbox{ where } F(x_1,\hdots, x_d)=\frac{1}{1+\gamma\sum_{i=1}^d x_i},
\end{equation}
with base case $x=1$.

To show the decay of correlations, one wants to show that after applying the recurrence starting from two different sets of values at $v_1,\dots,v_d$, the two computed values at $v$ will be ``closer'' than were the initial values at the $v_i$'s. This leads us to define a notion of distance. Often straightforward distances do not suffice to show decay of correlations, and distances defined via a ``potential'' function are used. We adapt this notion to the complex plane.

\subsection{Metrics for measuring the error in the complex plane}
We use a distance metric based on conformal density functions (see \cite{Conformal} for details).
\begin{definition}[Length, Distance, Metric]\label{def:ldm}
Let $U$ be a simply connected open subset of ${\mathbb C}$ and let
$\Phi:U\rightarrow {\mathbb R}_{>0}$ be a function
(called conformal density). The {\em length  with respect to $\Phi$}
of a path\footnote{\label{foot:path}Following \cite{Conformal}, paths are assumed to be continuous and piecewise continuously differentiable. By removing intervals where the derivative of the path is zero, we obtain a path whose derivative is non-zero a.e. (zeros of a continuous function cannot be dense in any interval unless the whole interval consists of zeros because of continuity).} $\eta:[0,1]\rightarrow U$  is defined as
\[\int_{0}^{1} \Phi(\eta(t)) \Big|\frac{\partial}{\partial
t}\eta(t)\Big| \, dt.\]
  The {\em distance with respect to $\Phi$} between two points $x,y\in
U$, denoted ${\rm dist}_{\Phi}(x,y)$, is the infimum of the lengths of
the paths $\eta$ connecting $x$ to $y$ (that is, $\eta(0)=x$
and $\eta(1)=y$).
We will refer to the metric induced by the distance function ${\rm
dist}_{\Phi}(\cdot,\cdot)$ as the (conformal) metric given by $\Phi$.
\end{definition}

Here is an example of a conformal metric.

\begin{example}\label{exal}
Suppose $U$ is the right-complex half-plane, that is $x\in U$ iff
$\mathrm{Re}(x)>0$, and
$\Phi(x)=1/\mathrm{Re}(x)$. The metric is the Poincar\'e metric in the
half-plane
(usually one takes the upper-complex half-plane) and the distance
between $a=x_1 + y_1 \mathrm{i}$
and $b=x_2 + y_2 \mathrm{i}$ is
\[2\ln\left(\frac{\sqrt{(y_2-y_1)^2 + (x_2-x_1)^2}+\sqrt{(y_2-y_1)^2 +
(x_1 + x_2)^2}}{2\sqrt{x_1 x_2}}\right).\]
\end{example}

We first quantify one-level correlation decay.
\begin{lemma}\label{lem:onelevelCD}
Let $U$ be a simply connected open subset of ${\mathbb C}$,
$\Phi:U\rightarrow {\mathbb R}_{> 0}$ be a conformal density function,
and ${\rm dist}_{\Phi}(\cdot,\cdot)$ be the metric given by $\Phi$.
Let $p$ and $q$ be conjugate exponents, that is, $1/p+1/q=1$, where
$p,q\in\mathbb{R}_{>0}\cup\{\infty\}$.

Suppose that $d\geq 1$ is an integer and $F:U^d\rightarrow U$ is a
holomorphic map. Let $x_1,\dots,x_d\in U$ and $y_1,\dots,y_d\in U$ and
let
$x=F(x_1,\dots,x_d)$ and $y=F(y_1,\dots,y_d)$. Assume that there
exists a real $\alpha\in(0,1)$ such that for any $z_1,\dots,z_d\in U$
\begin{equation}\label{netop}
\sum_{i=1}^d \Big|\Phi(F(z_1,\dots,z_d)) 
\frac{\partial F}{\partial z_i}(z_1,\dots,z_d)\frac{1}{\Phi(z_i)}\Big|^p \leq \alpha^p.
\end{equation}
Then
\begin{equation}\label{netop2}
{\rm dist}_{\Phi}(x,y)  \leq \alpha \left(\sum_{i=1}^d {\rm
dist}_{\Phi}(x_i,y_i)^q\right)^{1/q}.
\end{equation}
\end{lemma}
\begin{proof}
Let $\eps>0$. For $i\in [d]$, let $\eta_i$ be a path connecting
$x_i$ to $y_i$ of length
$\ell_i\leq {\rm dist}_{\Phi}(x_i,y_i) + \eps$. W.l.o.g., $\eta_i$ is
re-parameterized to uniform speed (cf. Footnote~\ref{foot:path}), 
that is, for a.e. $t\in [0,1]$ we have
\begin{equation}\label{eq:egb4y443f4fr}
\left|\frac{\partial}{\partial t}\eta_{i}(t)\right| \Phi(\eta_{i}(t))= \ell_i.
\end{equation}
We now define a path $\eta$ connecting $x$ to $y$:
\[\eta(t) := F(\eta_1(t),\dots,\eta_d(t)).\]
Let $L$ denote the length of $\eta$ and $F_i(x_1,\ldots,x_d)$ denote the function
$\frac{\partial F}{\partial x_i} (x_1,...,x_d)$. Then, using the triangle inequality and \eqref{eq:egb4y443f4fr}, we have
\begin{equation}\label{eq:b4tbt5by11212331}
\begin{aligned}
L&=\int_{0}^1 \Phi(\eta(t)) \Big|\frac{\partial}{\partial
t}\eta(t)\Big|\, dt = \int_{0}^1 \Phi(\eta(t)) \Big|\sum_{i=1}^d
 F_i(\eta_1(t),\dots,\eta_d(t))\,
\frac{\partial\eta_i }{\partial t} (t) \Big|\, dt \\
&\leq \int_{0}^1 \Phi(\eta(t)) \sum_{i=1}^d\Big|
 F_i(\eta_1(t),\dots,\eta_d(t))\,
\frac{\partial\eta_i }{\partial t} (t) \Big|\, dt \\
&=\int_{0}^1  \sum_{i=1}^d \Big|\Phi(\eta(t)) F_i (\eta_1(t),\dots,\eta_d(t))\frac{1}{\Phi(\eta_i(t))}\,
\ell_i\Big|\, dt.
\end{aligned}
\end{equation}
By H\"older's inequality and  condition~\eqref{netop}, for any $t\in
[0,1]$, we have
\begin{align*}
\sum_{i=1}^d \Big|\Phi(\eta(t)) & F_i(\eta_1(t),\dots,\eta_d(t))\frac{1}{\Phi(\eta_i(t))}\ell_i\Big|
\leq\\
&\left(\sum_{i=1}^d\Big|\Phi(\eta(t)) F_i(\eta_1(t),\dots,\eta_d(t))\frac{1}{\Phi(\eta_i(t))}\Big|^p\right)^{1/p}\left(\sum_{i=1}^d\ell_i^q\right)^{1/q}
\leq
\alpha\Big(\sum_{i=1}^d\ell_i^q\Big)^{1/q}.
\end{align*}
Integrating this for $t$ between 0 and 1 and combining with
\eqref{eq:b4tbt5by11212331}, we obtain
\begin{equation*}
{\rm dist}_{\Phi}(x,y) \leq L\leq  \alpha\Big(\sum_{i=1}^d\ell_i^q\Big)^{1/q}.
\end{equation*}
Taking $\eps\rightarrow 0$ we obtain
\begin{equation*}
{\rm dist}_{\Phi}(x,y) \leq \alpha \left(\sum_{i=1}^d {\rm
dist}_{\Phi}(x_i,y_i)^q\right)^{1/q}.\qedhere
\end{equation*}
\end{proof}

Now, given a rooted tree, our goal will be to bound the correlation decay at the root when we truncate the tree at depth $\Theta(\log n)$. Let $T$ be a finite tree rooted at a vertex $\rho$ and let $C$ be a subset of the leaves of $T$. Let $U\subseteq \mathbb{C}$. We will have a family of maps $\{F_d\}_{d\geq 1}$  where $F_d:U^d\mapsto U$ will be a symmetric map of arity $d$ (which will be the recurrence applied to a vertex of the tree with $d$ children).  Let $\sigma: C\rightarrow U$  be an arbitrary assignment of values in $U$ to the vertices of $C$. Let also $u_0\in U$ be the ``initial'' value ($u_0$ corresponds to the starting point of the recurrences). For a vertex $v$ in $T$ and an initial value $u_0\in U$, we define the quantity $r_v(C,\sigma, u_0)$ recursively as follows.  
 
\begin{equation} 
r_v(C,\sigma,u_0)=
\begin{cases}
 u_0&  \mbox{if $v$ is a leaf of $T$ and $v\notin C$},\\
\sigma(v)& \mbox{if } v\in C,\\
F_d(x_1,\hdots,x_d)&
\mbox{otherwise, where $x_i=r_{v_i}(C,\sigma,u_0)$}\\
 & \mbox{and   $v_1,\hdots,v_d$ are $v$'s children in $T$.}\end{cases}
\end{equation}
We can now study the sensitivity of $r_v(C,\sigma,u_0)$ to the assignment $\sigma$. The following lemma is the analogue of \cite[Lemma 3]{connconst1} for the complex plane and will be used to apply the correlation decay technique for graphs of bounded connective constant. 
\begin{lemma}\label{lem:vt41235evtvv}
Let $U$ be a simply connected open subset of ${\mathbb C}$ and 
$\Phi:U\rightarrow {\mathbb R}_{> 0}$ be a conformal density function.  For $d=1,2,\hdots,$ let $F_d:U^d\mapsto U$ be symmetric holomorphic maps. Suppose that there exists a real $\alpha \in (0,1)$ and conjugate exponents
$p$ and $q$ such that for every integer $d\geq 1$ and all $z_1,\hdots,z_d\in U$ it holds that  
\begin{equation}\label{eq:rv4th6y55323}
\sum_{i=1}^d \Big|\Phi(F_d(z_1,\dots,z_d)) \frac{\partial F_d}{\partial
z_i}(z_1,\dots,z_d)\frac{1}{\Phi(z_i)}\Big|^p \leq \alpha^p.
\end{equation}
Then, the following holds for any initial value $u_0\in U$ and any finite tree $T$ rooted at $\rho$.

Let $C$ be a subset of the leaves of $T$ and consider two arbitrary assignments $\sigma_1:C\rightarrow U$ and $\sigma_2:C\rightarrow U$. Then
\[|r_{\rho}(C,\sigma_1,u_0)-r_{\rho}(C,\sigma_2,u_0)| \leq \left(\frac{M}{L}\right)\bigg(\sum_{v\in C}\alpha^{q \cdot \depth(v)}\bigg)^{1/q},\]
where  $L:=\inf_{x\in U}\Phi(x)$, $M:=\max_{v\in C} \mathrm{dist}_{\Phi}(\sigma_1(v),\sigma_2(v))$ and $\depth(v)$ is the distance of $v$ from the root $\rho$.
\end{lemma}
\begin{proof}
The proof is close to  \cite[Proof of Lemma 3]{connconst1}, the only difference is we have to use the metric induced by $\Phi$, which we denote ${\rm dist}_{\Phi}(\cdot,\cdot)$.  In particular, for an arbitrary vertex $v$ in $T$, we use $C_v$ to denote the subset of $C$ that belongs to the subtree of $T$ rooted at $v$. Then, we will show that
\begin{equation}\label{eq:4rcrv4r23e234frf}
\big(\mathrm{dist}_{\Phi}(x,y)\big)^q \leq M^q\sum_{v'\in C_v}\alpha^{q(\depth(v')-\depth(v))}, \mbox{ where } x=r_v(C,\sigma_1,u_0) \mbox{ and } y=r_v(C,\sigma_2,u_0).
\end{equation}
We show this by induction. When $v$ is a leaf of $T$ and $v\notin C$, we have that $x=y=u_0$ and \eqref{eq:4rcrv4r23e234frf} holds trivially. When $v\in C$, then $x=\sigma_1(v)$, $y=\sigma_2(v)$ and $C_v=\{v\}$, so \eqref{eq:4rcrv4r23e234frf} holds by the definition of $M$. For the inductive case, suppose that $v$ neither is a leaf of $T$ nor belongs to $C$ and that  \eqref{eq:4rcrv4r23e234frf} holds for the children $v_1,\hdots,v_d$ of $v$.  For $i\in [d]$, set $x_i=r_{v_i}(C,\sigma_1,u_0)$, $y_i=r_{v_i}(C,\sigma_2,u_0)$ and observe that 
\[x=F_d(x_1,\hdots,x_d), \quad y=F_d(y_1,\hdots,y_d).\] 
By the inductive hypothesis, we also have that
\[\big(\mathrm{dist}_{\Phi}\big(x_i,y_i\big))^q \leq M^q\sum_{v'\in C_{v_i}}\alpha^{q(\depth(v')-\depth(v_i))}.\] 
From Lemma~\ref{lem:onelevelCD} and the assumption \eqref{eq:rv4th6y55323}, we obtain that
\begin{equation*}
\begin{split}
\mathrm{dist}_{\Phi}(x,y)^q \leq
\alpha^q \sum_{i=1}^d \mathrm{dist}_{\Phi}(x_i,y_i)^q
\leq \alpha^q \sum_{i=1}^d M^q \sum_{v'\in C_{v_i}}\alpha^{q(\depth(v')-\depth(v_i))}\\
= M^q \sum_{i=1}^d \sum_{u\in C_{v_i}}\alpha^{q(\depth(v')-\depth(v_i)+1)} =
M^q \sum_{v'\in C_{v}}\alpha^{q(\depth(v')-\depth(v))},
\end{split}
\end{equation*}
proving \eqref{eq:4rcrv4r23e234frf}. Notice that for any $x,y\in U$ we have
\[\mathrm{dist}_{\Phi}(x,y) = \int_{0}^{1} \Phi(\eta(t)) \Big|\frac{\partial}{\partial t}\eta(t)\Big| \, dt
\geq \int_{0}^{1} L \Big|\frac{\partial}{\partial t}\eta(t)\Big| \, dt
\geq L|x-y|.\]
The lemma follows from this and \eqref{eq:4rcrv4r23e234frf} (applied to $v=\rho$).
\end{proof}

\subsection{Applying the method for matchings}\label{sec:matchingmethod}

Suppose that $\gamma\in\mathbb{C}\setminus\mathbb{R}_{\leq 0}$. We will parameterise $\gamma$ as
\begin{equation}\label{eq:defQ}
\gamma=(1/Q)^2, \mbox{ where we choose $Q$ such that $\mathrm{Re}(Q)>0$}.
\end{equation}
Note that, in the choice of $Q$, we used the assumption that $\gamma$ is not a negative real number.\footnote{That is, 
we choose real numbers~$a$ and~$b$ and set
$Q=a+ \im b$. $a$ and $b$ are chosen  so that $\gamma=(1/Q)^2$ which implies 
$\gamma=1/(a^2-b^2+2a \im b)$. We can ensure $a>0$ by flipping the sign of~$b$, but not if $b=0$.} Let ${\cal H}$ be the
right complex half-plane, that is, the set of complex $x$ such that $\mathrm{Re}(x)>0$, and note that $Q\in{\cal H}$. We will also transform the space in which the quantities $p_v(G,\gamma)$ live using the map $x\mapsto x/Q$. In the transformed space, the recurrence \eqref{eq:xrecurrence} becomes
\begin{equation}\label{eq:yrecurrence}
y = F(y_1,\hdots,y_d) \mbox{ where }F(y_1,\hdots,y_d)=\frac{1}{Q + \sum_{i=1}^d y_i},
\end{equation}
where if $y$ corresponds to a leaf then $y=1/Q$ (we refer to this $y$ as the initial $y$). Let 
\begin{equation}\label{eq:defU}
U=\big\{y\in \mathbb{C}\mid \mathrm{Re}(y)>0, |y|< 1/\mathrm{Re}(Q)\big\}.
\end{equation}
The following lemma shows that the set $U$ is closed under application of the recurrence \eqref{eq:yrecurrence}.
\begin{lemma}\label{lem:loca}
Suppose that $y_1,\dots,y_d\in U$ and $\mathrm{Re}(Q)>0$.
Then, for $y$ given by~\eqref{eq:yrecurrence}, we have that $y\in U$ as well. In fact, we have that $\mathrm{Re}(y)\geq \frac{\mathrm{Re}(Q)}{\big(|Q|+\frac{d}{\mathrm{Re}(Q)}\big)^2}$.
\end{lemma}
\begin{proof}
Since $y_1,\hdots,y_d\in U$, we have that $\mathrm{Re}(y_1),\hdots,\mathrm{Re}(y_d)>0$. Using that $\mathrm{Re}(Q)>0$, we have that $\mathrm{Re}(Q + \sum_{i=1}^d y_i)=\mathrm{Re}(Q)+\sum_{i=1}^d \mathrm{Re}(y_i)>0$  and therefore $\mathrm{Re}(1/y)>0$. This yields that $\mathrm{Re}(y)>0$. Moreover, using again that $\mathrm{Re}(y_1),\hdots,\mathrm{Re}(y_d)>0$ and $\mathrm{Re}(Q)>0$, we have 
\begin{equation}
\label{eq:yboundedprf}
\Big|Q+\sum_{i=1}^d y_i\Big| >  \mathrm{Re}(Q),
\end{equation}
and hence $|y|< 1/\mathrm{Re}(Q)$. It follows that $y\in U$.

To prove the stronger bound on $\mathrm{Re}(y)$, note by the triangle inequality that  
\[\Big|Q+\sum_{i=1}^d y_i\Big|\leq |Q|+ \sum_{i=1}^d |y_i|\leq |Q|+\frac{d}{\mathrm{Re}(Q)},\]
and therefore
\[\mathrm{Re}(y) = \mathrm{Re}\Big(\frac{1}{Q+\sum_{i=1}^d y_i}\Big) =
\frac{\mathrm{Re}\left(Q+\sum_{i=1}^d y_i\right)}{|Q+\sum_{i=1}^d y_i|^2} \geq \frac{\mathrm{Re}(Q)}{\big(|Q|+\frac{d}{\mathrm{Re}(Q)}\big)^2}.\qedhere\]
\end{proof}

We next go on to show the required contraction properties for an appropriate function $\Phi$. This will largely be based on the following lemma from \cite{connconst2}.
\begin{lemma}\label{lem:Sinimport}
Let $\Delta$ and $\hat{\gamma}$ be positive real numbers. For $x\in (0,1]$, let $\hat{\Phi}(x)=\frac{1}{x(2-x)}$. Let $D=\max\{\Delta,\frac{3}{4\hat{\gamma}}\}$, $p=1/(1-\frac{1}{\sqrt{1+4\hat{\gamma} D}})$ and $q=p/(p-1)$. Then, for arbitrary $x_1,\hdots,x_d\in (0,1]$ it holds that
\[\bigg[\hat{\Phi}\bigg(\frac{1}{1+\hat{\gamma}\sum^d_{i=1} x_i}\bigg)\bigg]^p\sum^{d}_{i=1}\bigg[\frac{1}{\hat{\Phi}(x_i)}\frac{\hat{\gamma}}{(1+\hat{\gamma} \sum^d_{j=1} x_i)^2}\bigg]^p\leq \hat{\alpha}^p,\]
where $\hat{\alpha}=\frac{1}{D^{1/q}}\Big(1-\frac{2}{1+\sqrt{1+4\hat{\gamma} D}}\Big)$. 
\end{lemma}
\begin{proof}
The lemma follows from the derivations in~\cite{connconst2} as follows. Let $f_d(x_1,\dots,x_d)=\frac{1}{1+\hat\gamma\sum_{i=1}^d x_i}$ as defined in equation (2) of \cite{connconst2}. Then, Lemma 7 (see also Definition 10) of \cite{connconst2} shows that the expression
\begin{equation}
\label{eq:SSSYsym}
\hat\Phi(f_{d}(x_1,\dots,x_{d}))^p\sum_{i=1}^{d} \left(\frac{1}{\hat\Phi(x_i)}\left|\frac{\partial f_{d}}{\partial x_i}\right|\right)^p =
\hat\Phi\left(\frac{1}{1+\hat\gamma\sum_{i=1}^{d}x_i}\right)^p\sum_{i=1}^{d} \left[\frac{1}{\hat\Phi(x_i)}\frac{\hat\gamma}{(1+\hat\gamma\sum_{j=1}^{d} x_j)^2}\right]^p,
\end{equation}
constrained to $f_{d}(x_1,\dots,x_{d})=B$ for any fixed $B>0$, is maximized for $x_1=x_2\dots=x_{\hat{d}}=:x$ and $x_j=0$ for $j>d$, for some $\hat{d}\leq d$.

Then, \eqref{eq:SSSYsym} can be bounded from above by
\begin{equation}
\label{eq:SSSYnu}
\begin{split}
\hat\Phi\left(\frac{1}{1+\hat{d}\hat\gamma x}\right)^p \hat{d} \left[\frac{1}{\hat\Phi(x)}\frac{\hat\gamma}{(1+\hat\gamma \hat{d}x)^2}\right]^p =
\hat\Phi\left(\frac{1}{1+\hat{d}\hat\gamma x}\right)^p \hat{d} \left[\frac{1}{\hat{d}}\frac{1}{\hat\Phi(x)}\frac{\hat{d}\hat\gamma}{(1+\hat\gamma \hat{d}x)^2}\right]^p\\
 =
\hat\Phi\left(f_{\hat{d}}(x)\right)^p \hat{d} \left[\frac{1}{\hat{d}}\frac{1}{\hat\Phi(x)}|f_{\hat{d}}'(x)|\right]^p,
\end{split}
\end{equation}
where the univariate $f_{\hat{d}}(x):=\frac{1}{1+\hat{d}\hat\gamma x}$. From Lemma 9 (see also Definition 11) of \cite{connconst2} we get that for all $x$,
\begin{equation*}
\frac{1}{\hat{d}} \left[\hat\Phi(f_{\hat{d}}(x))\frac{|f_{\hat{d}}'(x)|}{\hat\Phi(x)}\right]^q \leq \frac{1}{D}\left(1-\frac{2}{1+\sqrt{1+4\hat\gamma D}}\right)^q,
\end{equation*}
where the left-hand side is maximized for $\hat{d}=D$. So,
\begin{equation*}
\left[\hat\Phi(f_{\hat{d}}(x))\frac{|f_{\hat{d}}'(x)|}{\hat\Phi(x)}\right]^q \leq
\frac{\hat{d}}{D}\left(1-\frac{2}{1+\sqrt{1+4\hat\gamma D}}\right)^q,
\end{equation*}
and, therefore,
\begin{equation*}
\hat\Phi(f_{\hat{d}}(x))\frac{|f_{\hat{d}}'(x)|}{\hat\Phi(x)} \leq \left(\frac{\hat{d}}{D}\right)^{1/q}
\left(1-\frac{2}{1+\sqrt{1+4\hat\gamma D}}\right).
\end{equation*}
Plugging this bound and the bound obtained in \eqref{eq:SSSYnu} into the expression from the lemma, we get
\begin{equation}
\label{eq:SSSYendpf}
\begin{split}
\bigg[\hat{\Phi}\bigg(\frac{1}{1+\hat{\gamma}\sum^d_{i=1} x_i}\bigg)\bigg]^p\sum^{d}_{i=1}\bigg[\frac{1}{\hat{\Phi}(x_i)}\frac{\hat{\gamma}}{(1+\hat{\gamma} \sum^d_{j=1} x_i)^2}\bigg]^p
\leq \hat\Phi(f_{\hat{d}}(x))^p \hat{d}\left[\frac{1}{\hat{d}}\frac{1}{\hat\Phi(x)}|f_{\hat{d}}'(x)|\right]^p\\
= {\hat{d}}^{1-p}\left[\hat\Phi(f_{\hat{d}}(x))\frac{1}{\hat\Phi(x)}|f_{\hat{d}}'(x)|\right]^p\leq {\hat{d}}^{1-p}\left(\frac{\hat{d}}{D}\right)^{p/q}\left(1-\frac{2}{1+\sqrt{1+4\hat\gamma D}}\right)^p.
\end{split}
\end{equation}
It remains to prove that the right-hand side is equal to $\hat\alpha^p$, which will finish the proof. To see this, we use $1/p+1/q=1$:
\begin{equation*}
\begin{split}
\hat\alpha^p = \left(\frac{1}{D^{1/q}}\left(1-\frac{2}{1+\sqrt{1+4\hat\gamma D}}\right)\right)^p
= \frac{1}{D^{p/q}}\left({\hat{d}}^{(1/p+1/q-1)}\left(1-\frac{2}{1+\sqrt{1+4\hat\gamma D}}\right)\right)^p\\
= \frac{{\hat{d}}^{(1-p+p/q)}}{D^{p/q}}\left(1-\frac{2}{1+\sqrt{1+4\hat\gamma D}}\right)^p,
\end{split}
\end{equation*}
which is equivalent to the right-hand side of \eqref{eq:SSSYendpf}.
\end{proof}

Using Lemma~\ref{lem:Sinimport}, we can obtain the following in the complex plane.
\begin{lemma}\label{lem:gggConnective}
Let $\Delta$ be a positive real number, $\gamma\in \Complex\setminus  \Reals_{< 0}$, and $Q,U$ be given from \eqref{eq:defQ} and \eqref{eq:defU}, respectively. Consider the function $\Phi:U\mapsto \mathbb{R}_{>0}$ given by $\Phi(y)=\frac{1}{\mathrm{Re}(y)(2/\mathrm{Re}(Q)-\mathrm{Re}(y))}$ and let
\begin{equation}\label{eq:12rfvbt4da}
\hat{\gamma}=\frac{2|\gamma|}{1+\cos(\emph{\arg} \gamma)},\ \, D=\max\{\Delta,\frac{3}{4\hat{\gamma}}\},\ \,  p=1/(1-\frac{1}{\sqrt{1+4\hat{\gamma} D}}),\ \, q=\frac{p}{p-1}, \ \, \alpha=\frac{1}{D^{1/q}}\Big(1-\frac{2}{1+\sqrt{1+4\hat{\gamma} D}}\Big).
\end{equation}
Then, the following holds for all integer $d\geq 1$.

Consider the map $F: U^d\mapsto U$ given by $F(y_1,\hdots,y_d)=\frac{1}{Q+\sum^d_{i=1}y_i}$. Then, for arbitrary $y_1,\dots,y_d\in U$ we have
\begin{equation}\label{eq:zz2}
\sum_{i=1}^d \Big|\Phi(F(y_1,\dots,y_d)) \frac{\partial F}{\partial y_i}(y_1,\dots,y_d)\frac{1}{\Phi(y_i)}\Big|^p \leq \alpha^p.
\end{equation}
\end{lemma}
\begin{proof}
Note that 
$\partial F(y_1,\ldots,y_d)/\partial y_i = -F(y_1,\ldots,y_d)^2$ and
$\Phi(y)=\frac{1}{|y|^2\mathrm{Re}(1/y)(\frac{2}{\mathrm{Re}(Q)}-\mathrm{Re}(y))}$, so we obtain that~\eqref{eq:zz2} is equivalent to
\begin{equation}\label{eq:hhhConnective}
\frac{1}{\bigg(\mathrm{Re}\left(Q+\sum_{i=1}^d y_i\right)\left(\frac{2}{\mathrm{Re}(Q)} - \mathrm{Re}\left(\frac{1}{Q + \sum_{i=1}^d y_i}\right) \right)\bigg)^p} \sum_{i=1}^d \Big[\mathrm{Re}(y_i)\Big( \frac{2}{\mathrm{Re}(Q)} - \mathrm{Re}(y_i) \Big)\Big]^p\leq \alpha^p.
\end{equation}
We have $\mathrm{Re}(1/z)\leq 1/\mathrm{Re}(z)$ and therefore 
\[\frac{2}{\mathrm{Re}(Q)} - \mathrm{Re}\bigg(\frac{1}{Q + \sum_{i=1}^d y_i}\bigg)\geq \frac{2}{\mathrm{Re}(Q)} - \frac{1}{\mathrm{Re}\left(Q+\sum_{i=1}^d y_i\right)}>0,\]
where the last inequality follows from the fact that $y_1,\dots,y_d\in {\cal H}$. Therefore, \eqref{eq:hhhConnective} will follow from
\begin{equation}\label{eq:hhh2Connective}
\frac{1}{\bigg(\mathrm{Re}\Big(Q+\sum_{i=1}^d y_i\Big)\Big(\frac{2}{\mathrm{Re}(Q)} - \frac{1}{\mathrm{Re}\left(Q+\sum_{i=1}^d y_i\right)} \Big)\bigg)^p} \sum_{i=1}^d \Big[\mathrm{Re}(y_i)\Big( \frac{2}{\mathrm{Re}(Q)} - \mathrm{Re}(y_i) \Big)\Big]^p\leq \alpha^p.
\end{equation}
For $\hat{\gamma}=\frac{2|\gamma|}{1+\cos(\arg \gamma)}=\frac{|\gamma|}{\cos^2(\frac{1}{2}\arg \gamma)}$, we have that $\hat{\gamma}= 1/(\mathrm{Re}(Q))^2$.  
We will show below that   \eqref{eq:hhh2Connective} is an immediate consequence of Lemma~\ref{lem:Sinimport} applied to $x_i=\mathrm{Re}(Q) \mathrm{Re}(y_i)$ and $\hat{\gamma}= 1/(\mathrm{Re}(Q))^2$.
First 
note that $x_1,\hdots,x_d\in (0,1]$ since $y_1,\hdots,y_d\in U$, so  Lemma~\ref{lem:Sinimport} indeed applies, showing  
\[
\sum^{d}_{i=1}
\bigg[\hat{\Phi}\bigg(\frac{1}{1+\hat{\gamma}\sum^d_{i=1} x_i}\bigg) \bigg(\frac{\hat{\gamma}}{\hat{\Phi}(x_i)}\bigg)
\bigg(
\frac{1}{(1+\hat{\gamma} \sum^d_{i=1} x_i)^2}
\bigg)
\bigg]^p\leq {\alpha}^p,\]
where $\hat{\Phi}(x)=\frac{1}{x(2-x)}$
so $\hat{\Phi}(1/x) = x^2/(2x-1)$. 
Using~$Y$ to denote
$\sum_{i=1}^d \mathrm{Re}(y_i)$ 
so that $\hat{\gamma}\sum_{i=1}^d x_i = Y/\mathrm{Re}(Q)$
and substituting in 
the values of $x_i$ and  $\hat{\gamma}$ this is 
 \[
\sum^{d}_{i=1}
\bigg[
\bigg(
\frac{{(1+ Y/\mathrm{Re}(Q))}^2}{2{(1+ Y/\mathrm{Re}(Q))}-1}
\bigg)
\bigg(
\frac{\mathrm{Re}(y_i)(2- \mathrm{Re}(Q)\mathrm{Re}(y_i))}{\mathrm{Re}(Q)}\bigg)
\bigg(
\frac{1}{(1+Y/\mathrm{Re}(Q))^2}
\bigg)
\bigg]^p\leq {\alpha}^p.\]
Cancelling terms and moving the sum inside, this is
  \[
  \bigg[
\frac{ 1}{2{(1+ Y/\mathrm{Re}(Q))}-1}
\bigg]^p
\sum^{d}_{i=1}
\bigg[
\bigg(
\frac{\mathrm{Re}(y_i)(2- \mathrm{Re}(Q)\mathrm{Re}(y_i))}{\mathrm{Re}(Q)}\bigg)
\bigg]^p\leq {\alpha}^p.\]
 To see that this is equivalent to
 \eqref{eq:hhh2Connective} 
 we need only show that the part outside of the sum is the same, e.g.,
 $$2{(1+ Y/\mathrm{Re}(Q))}-1 = 
 \mathrm{Re}\Big(Q+Y\Big)\Big(\frac{2}{\mathrm{Re}(Q)} - \frac{1}{\mathrm{Re}\left(Q+Y\right)} \Big).$$
 This follows easily since $\mathrm{Re}(Q+Y) = \mathrm{Re}(Q)+Y$.

\end{proof}

\subsection{Concluding the proof of Theorem~\ref{thm:connconst}}
\begin{thmconnconst}
\statethmconnconst
\end{thmconnconst}
\begin{proof}[Proof of Theorem~\ref{thm:connconst}]
If $\gamma$ is a non-negative real number, then the result follows from \cite[Theorem 1.3]{connconst2}. So we focus on the case where $\gamma$ is not real.

Using the telescoping expansion of $Z_G(\gamma)$ described in \eqref{eq:4bt45tb4tr899}, it suffices to give an algorithm that on an input graph $G\in \mathcal{F}_{\Delta,a,c}$, a vertex $v$ in $G$ and rational $\delta>0$ outputs in time $(\hat{c} n/\delta)^{O\big((1+a+\sqrt{|\hat{\gamma}| \Delta}) \log \Delta\big)}$ a quantity $\tilde{p}$ which satisfies $\tilde{p}=p_{v}(G,\gamma)\mathrm{e}^z$ for some complex number~$z$ with $|z| \leq \delta$.

Let $D,p,q,\alpha$ be the remaining constants in Lemma~\ref{lem:gggConnective} (other than $\hat{\gamma}$), i.e., 
\begin{equation*}\tag{\ref{eq:12rfvbt4da}}
D=\max\{\Delta,\frac{3}{4\hat{\gamma}}\},\ \,  p=\frac{1}{1-\frac{1}{\sqrt{1+4\hat{\gamma} D}}},\ \, q=\frac{p}{p-1}, \ \, \alpha=\frac{1}{D^{1/q}}\Big(1-\frac{2}{1+\sqrt{1+4\hat{\gamma} D}}\Big).
\end{equation*}
We will also use the parameterisation of Section~\ref{sec:matchingmethod}. Namely, as in \eqref{eq:defQ} and \eqref{eq:defU}, we will set 
\[\gamma=(1/Q)^2 \mbox{ and } U=\big\{y\in \mathbb{C}\mid \mathrm{Re}(y)>0,\, |y|< 1/\mathrm{Re}(Q)\big\},\]
where $Q$ is chosen so that $\mathrm{Re}(Q)>0$. Define also the constants 
\begin{equation}\label{eq:u0LM}
u_0:=1/Q,\quad L:= \frac{1}{2} (\mathrm{Re}(Q))^2 \mbox{ and } M:=\frac{2}{\mathrm{Re}(Q)}\Big(|Q|+\frac{n}{\mathrm{Re}(Q)}\Big)^2.
\end{equation}
Note that $u_0\in U$, since $\gamma$ is not real.

Let $T=T_{SAW}(v,G)$ be the self-avoiding walk tree rooted at $v$, then by Theorem~\ref{thm:selfavoid} we have that $p_{v}(G,\gamma)=p_v(T,\gamma)$, so it suffices to approximate $p_v(T,\gamma)$.  Let $C$ be the set of vertices in $T$ which are at distance $\ell$ from $v$, where 
$\ell$ is the smallest integer satisfying
\begin{equation}\label{eq:choiceell}
\ell\geq a\log n \mbox{ and }\frac{M}{L}\hat{c}^{1/q} (\Delta^{1/q}\alpha)^{\ell}\leq \frac{\mathrm{Re}(Q)}{2\big(|Q|+\frac{n}{\mathrm{Re}(Q)}\big)^2}\delta.
\end{equation}
Note that such an $\ell$ exists since $\Delta^{1/q}\alpha<1$ and, in fact, $\ell=O(\log(n/\delta))$. Let $T'$ be the subtree of $T$ obtained by deleting the descendants of $C$ (excluding vertices in $C$). We will show that
\begin{equation}\label{eq:goalrvwcrce}
p_v(T',\gamma)=p_v(T,\gamma)\mathrm{e}^{z}\mbox{ for some $|z|\leq\delta$.}\end{equation}
From this, it follows that we can just output $\tilde{p}=p_v(T',\gamma)$ as an approximation to  $p_{v}(G,\gamma)=p_v(T,\gamma)$. Since $G\in\mathcal{F}_{\Delta,a,c}$ and $\ell\geq a\log n$, we have that $T'$ is a tree with at most $c\Delta^{\ell}$ vertices, and hence we can compute $p_v(T',\gamma)$ in time $(\hat{c} n/\delta)^{O\big((1+a+\sqrt{|\hat{\gamma}| \Delta}) \log \Delta\big)}$.

It therefore remains to prove \eqref{eq:goalrvwcrce}. For a graph $H$ and a vertex $w$ we define $\hat{p}_w(H,\gamma)$ by
\[\hat{p}_w(H,\gamma)=p_w(H,\gamma)/Q.\]
Let $\{F_d\}_{d\geq 1}$ be the sequence of maps corresponding to recurrence in \eqref{eq:yrecurrence}, i.e., for  integer $d\geq1$ and $y_1,\hdots,y_d\in U$
\[F_d(y_1,\hdots,y_d)=\frac{1}{Q+\sum^d_{i=1}y_i}.\]
For a vertex $w$ in $T'$ and an assignment $\sigma: C\rightarrow U$, we define the quantity $r_w(\cdot,\cdot,\cdot)$
 
\begin{equation} 
r_w(C,\sigma,u_0)=\left\{\begin{array}{ll} u_0&  \mbox{if $w$ is a leaf of $T'$ that is not in $C$}\\
\sigma(v)& \mbox{if } w\in C\\
F_d(x_1,\hdots,x_d)& \mbox{where  $w_1,\hdots,w_d$ are $w$'s children in $T'$ and $x_i=r_{w_i}(C,\sigma,u_0)$.}\end{array}\right.
\end{equation} 
Let $\sigma_1,\sigma_2$ be assignments on $C$ such that 
\[\mbox{for all $w\in C$, $\sigma_1(w)=u_0$ and $\sigma_2(w)=\hat{p}_w(T_w,\gamma)$},\]
where $T_w$ denotes the subtree of $T$ induced by the descendants of  $w$ (including $w$). Note that 
\begin{equation}\label{eq:grbetb774546}
\hat{p}_v(T',\gamma)=r_v(C,\sigma_1,u_0), \quad \hat{p}_v(T,\gamma)=r_v(C,\sigma_2,u_0).
\end{equation}
Moreover, by Lemma~\ref{lem:gggConnective}, we have that the family of maps $\{F_d\}_{d\geq 1}$ and the function $\Phi:U\rightarrow {\mathbb R}_{> 0}$ given by $\Phi(y)=\frac{1}{\mathrm{Re}(y)\left(\frac{2}{\mathrm{Re}(Q)} - \mathrm{Re}(y)\right)}$ satisfy the hypotheses of Lemma~\ref{lem:vt41235evtvv}. We will also show shortly that the constants $L,M$ defined in \eqref{eq:u0LM} satisfy 
\begin{equation}\label{eq:ML}
L\leq \inf_{x\in U}\Phi(x)\quad \mbox{and}\quad M\geq \max_{w\in C} \mathrm{dist}_{\Phi}(\sigma_1(w),\sigma_2(w)),
\end{equation}
where $\mathrm{dist}_{\Phi}(\sigma_1(w),\sigma_2(w))$ is the metric corresponding to $\Phi$ (cf. Definition~\ref{def:ldm}). Let us assume \eqref{eq:ML} for the moment and conclude the proof of the Theorem. Applying the conclusion of Lemma~\ref{lem:vt41235evtvv} to the tree $T'$, we obtain that
\begin{equation}\label{eq:rvrv12rvrv}
\begin{aligned}
|r_v(C,\sigma_1,u_0)-r_v(C,\sigma_2,u_0)|&\leq \frac{M}{L}\Big(\sum_{w\in C} \alpha^{q\cdot \depth(w)}\Big)^{1/q}= \frac{M}{L}|C|^{1/q} \alpha^{\ell}\leq \frac{M}{L}c^{1/q} (\Delta^{1/q}\alpha)^{\ell}\\
&\leq \frac{\mathrm{Re}(Q)}{2\big(|Q|+\frac{n}{\mathrm{Re}(Q)}\big)^2}\delta,
\end{aligned}
\end{equation}
where we used that $G\in \mathcal{F}_{\Delta,a,c}$ and the choice of $\ell$ in \eqref{eq:choiceell}. To prove \eqref{eq:ML}, note that for $y\in U$, we have that $\mathrm{Re}(y)\leq |y| \leq \frac{1}{\mathrm{Re}(Q)}$ and hence, using also the fact that $\mathrm{Re}(y)>0$, it follows that
\[
\inf_{y\in U}\Phi(y) = \inf_{y\in U} \frac{1}{\mathrm{Re}(y)\big(\frac{2}{\mathrm{Re}(Q)} - \mathrm{Re}(y)\big)}
\geq \frac{1}{2}(\mathrm{Re}(Q)^2)=L.\]
To prove the bound on $M$ in \eqref{eq:ML}, let us consider arbitrary  $w\in C$, we will show that $\mathrm{dist}_{\Phi}(\sigma_1(w),\sigma_2(w))\leq M$.  We have  $\sigma_1(w)=u_0\in U$ and
\[\mathrm{Re}(\sigma_1(w))=\mathrm{Re}(u_0)=\mathrm{Re}(1/Q)=\mathrm{Re}(Q)/|Q|^2.\]
For $\sigma_2(w)$, let $w_1,\hdots, w_d$ be the children of $w$ in the tree $T$ and note that $d\leq n$. Set $y_i=\hat{p}_{w_i}(T_{w_i},\gamma)$ so that $y_1,\hdots, y_d\in U$ and  $\sigma_2(w)=F(y_1,\hdots,y_d)$. It follows from Lemma~\ref{lem:loca} that $\sigma_2(w)\in U$ and
\[\mathrm{Re}(\sigma_2(w))\geq \frac{\mathrm{Re}(Q)}{\big(|Q|+\frac{d}{\mathrm{Re}(Q)}\big)^2}\geq \frac{\mathrm{Re}(Q)}{\big(|Q|+\frac{n}{\mathrm{Re}(Q)}\big)^2}.\] 
Consider a path $\eta$ given by $\eta(t)=(1-t)\sigma_1(w)+t \sigma_2(w)$ for $t\in [0,1]$.  Then, for all $t\in [0,1]$, we have that 
\[\mathrm{Re}(\eta(t))\geq \min\big\{|\mathrm{Re}(\sigma_1(w))|, |\mathrm{Re}(\sigma_2(w))|\big\}\geq \frac{\mathrm{Re}(Q)}{\big(|Q|+\frac{n}{\mathrm{Re}(Q)}\big)^2}\]
and hence
\[\Phi(\eta(t)) = \frac{1}{\mathrm{Re}(\eta(t))\left(\frac{2}{\mathrm{Re}(Q)} - \mathrm{Re}(\eta(t))\right)} \leq \Big(|Q|+\frac{n}{\mathrm{Re}(Q)}\Big)^2.\]
It follows that
\begin{align*}
\mathrm{dist}_{\Phi}(\sigma_1(w),\sigma_2(w))&\leq \int_{0}^{1} \Phi(\eta(t)) \Big|\frac{\partial}{\partial t}\eta(t)\Big| \, dt
\leq \Big(|Q|+\frac{n}{\mathrm{Re}(Q)}\Big)^2\big|\sigma_1(w)-\sigma_2(w)\big|\\
&\leq\frac{2}{\mathrm{Re}(Q)}\Big(|Q|+\frac{n}{\mathrm{Re}(Q)}\Big)^2,
\end{align*}
where the last inequality is obtained from the triangle inequality and $\sigma_1(w),\sigma_2(w)\in U$. This finishes the proof of \eqref{eq:ML}.

Now, from Lemma~\ref{lem:loca} we have that $|r_v(C,\sigma_2,u_0)|\geq \mathrm{Re}(r_v(C,\sigma_2,u_0))\geq \frac{\mathrm{Re}(Q)}{\big(|Q|+\frac{n}{\mathrm{Re}(Q)}\big)^2}$, and therefore \eqref{eq:rvrv12rvrv} combined with  \eqref{eq:grbetb774546} gives that 
\begin{equation*}
\frac{|p_v(T,\gamma)-p_v(T',\gamma)|}{|p_v(T,\gamma)|}\leq \delta/2.
\end{equation*}
This yields \eqref{eq:goalrvwcrce}, therefore completing the proof of Theorem~\ref{thm:connconst}.
\end{proof}

\section{Proof of Theorems~\ref{thm:hard} and~\ref{thm:hardsign} (bounded degree)}\label{sec:maxhard}

\subsection{Implementing Edge Activities}
Let $G=(V,E)$ be a graph and $u,v\in V$. Analogously to the notation $Z_{G,u}(\gamma)$ and $Z_{G,\neg u}(\gamma)$ of Section~\ref{sec:prelims}, we will denote
\begin{align*}
Z_{G,u,v}(\gamma)&:=\sum_{M\in \mathcal{M}_G;\, u,v\in \mathsf{ver}(M)}\gamma^{|M|},& Z_{G,u,\neg v}(\gamma)&:=\sum_{M\in \mathcal{M}_G;\, u\in \mathsf{ver}(M), v\notin \mathsf{ver}(M)}\gamma^{|M|},\\
Z_{G,\neg u,\neg v}(\gamma)&:=\sum_{M\in \mathcal{M}_G;\, u,v\notin \mathsf{ver}(M)}\gamma^{|M|}, & Z_{G,\neg u, v}(\gamma)&:=\sum_{M\in \mathcal{M}_G;\, u\notin \mathsf{ver}(M), v\in \mathsf{ver}(M)}\gamma^{|M|}.
\end{align*}
Thus, $Z_{G,u,v}(\gamma)$ is the contribution to the partition function $Z_G(\gamma)$ from those matchings $M\in \mathcal{M}_G$ such that both $u,v$ are matched in $M$, while $Z_{G,\neg u,\neg v}(\gamma)$ is the contribution to the partition function $Z_G(\gamma)$ from those matchings $M\in \mathcal{M}_G$ such that neither of $u,v$ are matched in $M$. 
\begin{definition}\label{def:edgeimpl}
Fix a real number $\gamma$. Given $\gamma$,
the graph $G=(V,E)$  is said to \emph{implement} the edge activity $\gamma'\in \Reals$ with \emph{accuracy} $\epsilon>0$ if there are vertices $u$, $v$ in $G$ such that  $Z_{G, \neg u,\neg v}(\gamma)\neq 0$  and
\begin{enumerate}
\item \label{it:uvone} $u,v$ have degree one in $G$ and $(u,v)\notin E$, 
\item \label{it:bothmatched} $\displaystyle \Big|\frac{Z_{G,u,\neg v}(\gamma)}{Z_{G,\neg u,\neg v}(\gamma)}\Big|\leq \epsilon$, $\displaystyle \Big|\frac{Z_{G,\neg u,v}(\gamma)}{Z_{G,\neg u,\neg v}(\gamma)}\Big|\leq \epsilon$,
\item  \label{it:ps123}   $\displaystyle \Big|\frac{Z_{G,u,v}(\gamma)}{Z_{G,\neg u,\neg v}(\gamma)}-\gamma'\Big|\leq \epsilon $.
\end{enumerate}
We call $u,v$ the terminals of $G$. If both of Items~\ref{it:bothmatched} and~\ref{it:ps123} hold with $\epsilon=0$, we say that $G$ implements the edge activity $\gamma'$ (perfectly).
\end{definition}

\begin{definition}\label{def:sizes}
Let $\alpha$ be a rational number and write $\alpha=p/q$, where  $p,q$ are integers such that $\mathrm{gcd}(p,q)=1$. Then, the \emph{size} of $\alpha$, denoted by $\size{\alpha}$, is given by $1+\log(|p|+|q|)$. For $\alpha_1,\hdots,\alpha_t\in \mathbb{Q}$, we denote by $\size{\alpha_1,\hdots,\alpha_t}$ the total of the sizes of $\alpha_1,\hdots,\alpha_t$.

For a multivariate polynomial $P(x_1,\hdots, x_n)$ of degree $d$ with rational coefficients $\alpha_1,\hdots, \alpha_t$, we let $\size{P}$ be $d+\size{\alpha_1,\hdots,\alpha_t}$. 
\end{definition}

\newcommand{\statelemexpmatch}{Let $\Delta\geq 3$ be an integer and $\gamma<-\frac{1}{4(\Delta-1)}$ be a rational number.

There is an algorithm which, on input rational $\gamma'\leq 0$ and $\eps>0$, outputs in $poly(\size{\gamma',\eps})$ time a bipartite graph $G$ of maximum degree at most  $\Delta$ with terminals $u,v$ in the same part of the vertex partition of $G$ so that $G$ implements $\gamma'$ with accuracy $\eps$.}
\begin{lemma}\label{lem:expmatch}
\statelemexpmatch
\end{lemma}

\subsection{Proof of Main Hardness Results}
In this section, assuming our key Lemma~\ref{lem:expmatch} (which will be proved in Section~\ref{sec:keylemmahard}), we complete the proof of Theorems~\ref{thm:hard} and~\ref{thm:hardsign}. To capture the restriction on the input graph $G$ within those theorems and also for easy reference within the proofs, it will be convenient to define the following computational problems  which capture the problems of multiplicatively approximating the norm of $Z_G(\gamma)$ and determining the sign of $Z_G(\gamma)$.

\prob{
$\Matchings{K}$.} 
{ A bipartite graph $G$ with maximum degree at
most $\Delta$.} 
 { If $Z_G(\gamma)=0$  then the algorithm may output any rational number. Otherwise,
 it must output a  rational number $\widehat{N}$ such tha 
$\widehat{N}/K \leq
|Z_{G}(\gamma)|\leq K \widehat{N}$.}
\prob{$\SignMatchings$.} 
 { A bipartite graph $G$ with maximum degree at
most $\Delta$.}
{If $Z_G(\gamma)=0$ then the output may be either $+$ or $-$. Otherwise, the output is $\mathrm{sign}(Z_G(\gamma))$.
} 

Using this language, we restate Theorems~\ref{thm:hard} and~\ref{thm:hardsign} which we focus on proving next.
\begin{thmhard}
Let $\Delta\geq 3$ and $\gamma<-\frac{1}{4(\Delta-1)}$ be a rational number. Then,   $\Matchings{1.01}$  is $\numP$-hard. 
\end{thmhard}
\begin{thmhardsign}
Let $\Delta\geq 3$ and $\gamma<-\frac{1}{4(\Delta-1)}$ be a rational number. Then, 
 $\SignMatchings$  is $\numP$-hard. 
\end{thmhardsign}
\begin{proof}[Proof of Theorems~\ref{thm:hard} and~\ref{thm:hardsign}]
Let $\gamma_0=-1/10$ and $\mathcal{G}$ be the set of graphs of maximum degree 3. It is well-known \cite[Theorem 3]{Cai2012} that the problem of computing $Z_G(\gamma_0)$ on input a graph $G\in \mathcal{G}$ is \#P-hard.  Moreover, by Corollary~\ref{lem:positive} we have that $Z_G(\gamma_0)>0$ for all graphs $G\in \mathcal{G}$. 

Using an oracle for either $\Matchings{1.01}$ or $\SignMatchings$, we will design a polynomial time algorithm to compute the ratio $\frac{Z_G(\gamma_0)}{Z_{G-e^*}(\gamma_0)}$ for an arbitrary graph $G\in \mathcal{G}$ and an arbitrary edge $e^*$ of $G$; note that this ratio is well-defined since $Z_{G-e^*}(\gamma_0)> 0$. With such a subroutine at hand, it is standard to compute $Z_G(\gamma_0)$ using self-reducibility techniques\footnote{Namely, let $e_1, e_2, \hdots, e_m$ be an enumeration of the edges of $G$ and let $G_i$ be the graph where the edges $e_i,\hdots, e_m$ are deleted (note that $G_{m+1}=G$ and $G_1$ is the empty graph). Then, we have that $Z_G(\gamma_0)=\prod^{m}_{i=1} \frac{Z_{G_{i+1}}(\gamma_0)}{Z_{G_{i}}(\gamma_0)}$.}, which therefore proves that  $\Matchings{K}$ and $\SignMatchings$  are both \#P-hard.

Let $G=(V,E)$ be an arbitrary  graph of maximum degree $3$ with $|V|=n$ and $|E|=m$.  Let also $e^*=(u^*,v^*)$ be an arbitrary edge of $G$. Our goal is to compute $\frac{Z_G(\gamma_0)}{Z_{G-e^*}(\gamma_0)}$. To do this, let 
\[\alpha=Z_{G\backslash\{u^*,v^*\}}(\gamma_0),\quad \beta=Z_{G-e^*}(\gamma_0)\]
and note that
\[\frac{Z_G(\gamma_0)}{Z_{G-e^*}(\gamma_0)}=\frac{\gamma_0 Z_{G\backslash\{u^*,v^*\}}(\gamma_0)+Z_{G-e^*}(\gamma_0)}{Z_{G-e^*}(\gamma_0)}=\gamma_0\frac{\alpha}{\beta}+1,\]
so it suffices to compute $R_{\mathsf{goal}}:=-\frac{\beta}{\alpha}$.  Note that $R_{\mathsf{goal}}$ is a (well-defined) negative number since both $\alpha=Z_{G\backslash\{u^*,v^*\}}(\gamma_0)$ and $\beta=Z_{G-e^*}(\gamma_0)$ are positive. Moreover, since $\gamma_0=-1/10$ and each of the graphs $G-e^*$ and $G\backslash\{u^*,v^*\}$ have at most $m\leq 2n$ edges, we have that $\alpha=P/10^{2n}$ and $\beta=Q/10^{2n}$ for some  integers $P,Q$ satisfying the crude bounds $1\leq P, Q\leq 20^{2n}$. It follows that $R_{\mathsf{goal}}\in \mathcal{R}_0$ where
\[\mathcal{R}_0:=\big\{-p/q\mid 0\leq p,q\leq 20^{2n}, q\neq 0 \big\}.\]
Let $N:=10^3n$, $L_0:=-20^{2n}$ and $U_0:=0$. For $i=1,\hdots, N$, we will  show a recursive procedure to compute rationals $L_{i}, U_i$ such that 
\begin{equation}\label{eq:binsearch}
R_{\mathsf{goal}}\in [L_i,U_i], \quad |U_i-L_i|\leq \frac{7}{8}|U_{i-1}-L_{i-1}|, \quad L_i, U_i \in \mathcal{R}_i:=\big\{-p/q\mid 0\leq p,q\leq 8^i20^{2n}, q\neq 0 \big\}.
\end{equation}
Observe that any two distinct rationals in $\mathcal{R}_0$ differ\footnote{\label{foot:distinct}Suppose that $\alpha=p/q$ and $\alpha'=p'/q'$ are distinct rationals, where  $p,q,p',q'$ are integers whose absolute values are all less than $M$ for some $M>0$. Then,  we have that $|\alpha-\alpha'|=\frac{|pq'-qp'|}{|qq'|}\geq 1/M^2$ since $pq'-qp'$ is an integer distinct from zero.}  by at least $1/400^{2n}$. 
Since $U_N-L_N\leq (7/8)^{N}{20^{2n}}< 1/400^{2n}$, it follows that using the values of $L_N, U_N$ we can in fact figure out the exact value of $R_{\mathsf{goal}}$ in $poly(n)$ time (see \cite[Footnote 8]{buys2020leeyang}).  

Let $\epsilon':=1/20^{8(n+N)}$. For a number $R\in \mathcal{R}_N$,  let  $G_R$ be the graph obtained as follows. First, using the algorithm in Lemma~\ref{lem:expmatch}, we construct in $poly(n)$ time a bipartite graph   $H_0$ of maximum degree $\Delta$ with terminals $u_0$ and $v_0$ 
(in the same part of the vertex partition)
which implements the activity $\gamma_0=-1/10$ with accuracy $\epsilon=\frac{\epsilon'}{5^{4n}\max\{|\gamma_0|,|R|\}}$. Similarly, we also construct  in $poly(n)$ time a bipartite graph $H_1$ of maximum degree $\Delta$ with terminals $u_1$ and $v_1$ 
(in the same part of the vertex partition)
which implements the activity $R$ with accuracy $\epsilon$.
For every edge $e\in E$ such that $e\neq e^*$, let $H^{(e)}$ be a copy of $H_0$ and set $\gamma^{(e)}=\gamma_0$. For $e=e^*$, let $H^{(e)}$ be a copy of $H_1$ and set $\gamma^{(e)}=R$. For $e\in E$, we also denote by $u^{(e)}$ and $v^{(e)}$ the terminals of $H^{(e)}$ and set\footnote{Note that $x^{(e)}_R(\cdot,\cdot)$  depends on $H^{(e)}$, which in turn depends on $R$ (via the choice of $\epsilon$). The reason that we explicitly note the dependence on $R$ but not on $H^{(e)}$ is for convenience.} 
\begin{align*}
x^{(e)}_{R}(1,1)=Z_{H^{(e)},u^{(e)},v^{(e)}}(\gamma), \quad x^{(e)}_{R}(0,0)=Z_{H^{(e)},\neg u^{(e)},\neg v^{(e)}}(\gamma),\\
x^{(e)}_{R}(1,0)=Z_{H^{(e)},u^{(e)},\neg v^{(e)}}(\gamma), \quad x^{(e)}_{R}(0,1)=Z_{H^{(e)},\neg u^{(e)}, v^{(e)}}(\gamma).
\end{align*}
Note that for all $e\in E$, we have that 
\begin{equation}\label{eq:tv4tv444}
x^{(e)}_{R}(0,0)\neq 0, \quad \Bigg|\frac{x^{(e)}_{R}(1,0)}{x^{(e)}_{R}(0,0)}\Bigg|\leq \epsilon, \quad \Bigg|\frac{x^{(e)}_{R}(0,1)}{x^{(e)}_{R}(0,0)}\Bigg|\leq \epsilon, \mbox{\ \  and\ \ }\bigg|\frac{x^{(e)}_{R}(1,1)}{x^{(e)}_{R}(0,0)}-\gamma^{(e)}\bigg|\leq \epsilon.
\end{equation} 
Let $G_R$ be the bipartite graph obtained from $G$ by replacing every edge $e\in E$ with the graph $H^{(e)}$ and  identifying the endpoints of $e$ with the terminals $u^{(e)}$ and $v^{(e)}$ (bipartiteness of  $G_R$ follows from the fact that $u^{(e)}$ and $v^{(e)}$ lie in the same part of the vertex partition of $H^{(e)}$). Also, let $T_R$ be the bipartite graph obtained from $G$ by replacing every edge $e\in E$ with the graph $H^{(e)}$ and  deleting the terminals $u^{(e)}$ and $v^{(e)}$ (so $T_R$ is a disjoint union of copies of $H_0$ and $H_1$ with all the terminal vertices deleted). We have that 
\begin{equation}\label{eq:3rf34f312}
Z_{T_R}(\gamma)=\prod_{e\in E}x^{(e)}_{R}(0,0)\neq 0.
\end{equation}
We will also show that  
\begin{equation}\label{eq:rtv4tv12}
\big|\alpha R +\beta-f_R\big|\leq \epsilon',\mbox{ where }f_R:=\frac{Z_{G_R}(\gamma)}{Z_{T_R}(\gamma)}.
\end{equation}
Assuming \eqref{eq:rtv4tv12} for the moment, we first conclude the reductions for $\Matchings{1.01}$ and $\SignMatchings$ using the binary search technique of \cite{ComplexIsing}. Note that if $R_1,R_2$ are distinct numbers in $\mathcal{R}_N$ with $R_1>R_2$, then we have from footnote~\ref{foot:distinct} that $R_1-R_2\geq 1/(8^{2N} 20^{4n})$   and therefore 
\begin{equation}\label{eq:r3rfrf21}
f_{R_1}-f_{R_2}\geq \alpha(R_1-R_2)-2\epsilon' \geq 1/(8^{2N} 20^{4n}20^{2n})-2\epsilon'\geq 2\epsilon'>0.
\end{equation}
From \eqref{eq:rtv4tv12} with $R=R_{\mathsf{goal}}$, we obtain that $|f_{R_{\mathsf{goal}}}|\leq \epsilon'$ and therefore, using also \eqref{eq:r3rfrf21}, 
\begin{equation}\label{eq:r3d2d22w}
\mbox{for $R\in \mathcal{R}_N$ it holds that $f_R>0$ if $R>R_{\mathsf{goal}}$ and $f_R<0$ if $R<R_{\mathsf{goal}}$.}
\end{equation} 

\textbf{\#P-hardness of $\Matchings{1.01}$:} Assume that for some $i\in \{1\hdots, N\}$ we have computed $L_{i-1},U_{i-1}\in\mathcal{R}_{i-1}$ such that $R_{\mathsf{goal}}\in [L_{i-1},U_{i-1}]$.   We will show how to compute $L_{i}, U_{i}$ satisfying \eqref{eq:binsearch}. 

Let $\ell:=(U_{i-1}-L_{i-1})/8$. For $j=0,\hdots, 8$, let $R_{j}=L_{i-1}+j \ell$ so that $R_0=L_{i-1}$ and $R_8=U_{i-1}$. Using the oracle to $\Matchings{1.01}$ on inputs $T_{R_j}$ and $G_{R_j}$ (note that $Z_{T_{R_j}}(\gamma)\neq 0$ from \eqref{eq:3rf34f312} and $Z_{G_{R_j}}(\gamma)\neq 0$, unless perhaps in the case where  $R_j= R_{\mathsf{goal}}$ which can happen for at most one index $j$), we can estimate the partition functions $Z_{G_{R_j}}(\gamma)$ and $Z_{T_{R_j}}(\gamma)$  within a factor of 1.01 and therefore we can compute $\hat{f}_{R_j}$ such that
\begin{equation}\label{eq:g45b3f4}
 (1-\eta) |{f}_{R_j}|\leq|\hat{f}_{R_j}|\leq (1+\eta) |{f}_{R_j}| \mbox{ where } \eta:=0.05.
\end{equation}
Suppose that $j\in\{0,1,\hdots, 7\}$ is an index such that $R_{\mathsf{goal}}<R_{j}$. Then, we have that $\alpha R_{j+1}+\beta>\alpha R_j+\beta>0$, so using \eqref{eq:g45b3f4} and \eqref{eq:rtv4tv12} we  obtain that
\begin{align*}
|\hat{f}_{R_{j+1}}|-|\hat{f}_{R_{j}}|&\geq (1-\eta)|{f}_{R_{j+1}}|-(1+\eta)|{f}_{R_{j}}|\\
&\geq  (1-\eta)(\alpha R_{j+1}+\beta-\epsilon')-(1+\eta)(\alpha R_{j}+\beta+\epsilon')\\
&=\alpha(R_{j+1}-R_{j})-\eta\big(\alpha(R_{j+1}+R_{j})+2\beta\big)-2\epsilon'\\
&=\alpha\big(R_{j+1}-R_{j}-\eta(R_{j+1}+R_{j}-2R_{\mathsf{goal}})\big)-2\epsilon'\\
&\geq \alpha \big(\ell-16\eta\ell\big)-2\epsilon'\geq \alpha \ell/10-2\epsilon'> 0,
\end{align*}
where the second equality follows from the fact $\alpha R_{\mathsf{goal}}+\beta=0$ and the last inequality follows from $\alpha\geq 1/10^{2n}$, $\ell\geq 1/8^{2(N+1)} 20^{4n}$ and $\epsilon'=1/20^{8(n+N)}$.
An analogous calculation shows that if $j\in\{1,\hdots, 8\}$ is an index  such that $R_{j}<R_{\mathsf{goal}}$ then $|\hat{f}_{R_{j-1}}|>|\hat{f}_{R_{j}}|$. Therefore, at least one of the following series of inequalities holds: either $|\hat{f}_{R_0}|>|\hat{f}_{R_1}|>|\hat{f}_{R_2}|>|\hat{f}_{R_3}|$ or $|\hat{f}_{R_5}|<|\hat{f}_{R_6}|<|\hat{f}_{R_7}|<|\hat{f}_{R_8}|$ (or both). In the first case, we can be sure that $R_{\mathsf{goal}}\notin [R_0,R_1]$ 
(since $R_{\mathsf{goal}}< R_2$ would imply $|\hat{f}_{R_3}|> |\hat{f}_{R_2}|$)
and therefore we can set $L_{i}=R_1, U_{i}=R_8$. In the second case, we can be sure that $R_{\mathsf{goal}}\notin [R_7,R_8]$ 
(since $R_6 < R_{\mathsf{goal}} $ would imply  $|\hat{f}_{R_{5}}|>|\hat{f}_{R_{6}}|$ )
and therefore we can set $L_{i}=R_0, U_{i}=R_7$. In both cases, we have that $L_i, U_i$ satisfy \eqref{eq:binsearch} as wanted.

\textbf{\#P-hardness of $\SignMatchings$:} This is analogous to the previous reduction, only easier
(so essentially the binary search follows the simpler method of~\cite{SignTutte}). Following the setting above, to compute $L_i, U_i$ we use the oracle to $\SignMatchings$ on inputs $T_{R_j}$ and $G_{R_j}$ to decide whether $f_{R_j}>0$ (as noted before, $Z_{T_{R_j}}(\gamma)\neq 0$ from \eqref{eq:3rf34f312} and $Z_{G_{R_j}}(\gamma)= 0$ is possible for at most one index $j$). Using \eqref{eq:r3d2d22w}, we obtain that either $\hat{f}_{R_0},\hat{f}_{R_1},\hat{f}_{R_2},\hat{f}_{R_3}<0$ or $\hat{f}_{R_5},\hat{f}_{R_6},\hat{f}_{R_7},\hat{f}_{R_8}>0$ (or both). In the former case we can then set $L_{i}=R_1, U_{i}=R_8$  and in the latter case we can set $L_{i}=R_0, U_{i}=R_7$;  in both cases, we have that $L_i, U_i$ satisfy \eqref{eq:binsearch}.

To finish the proof, it remains to establish \eqref{eq:rtv4tv12}.  We will need some definitions. For a matching $M_R$ of $G_R$, the phase of the matching $M_R$, denoted by $\mathcal{Y}(M_R)$, is a 0-1 vector indexed by pairs $(e,v)$ such that $e$ is an edge of $G$ and $v$ is an endpoint of $e$; we set the $(e,v)$ entry of $\mathcal{Y}(M_R)$ equal to 1 if $v$ is matched in the matching $M_{R}$ by edges in the gadget $H^{(e)}$. 

Let $\mathcal{P}$ be the set  consisting of all phases (i.e. the values of $\mathcal{Y}(M_R)$ as $M_R$ ranges over matchings of $G_R$). For a vertex $v\in V$ denote by $d_{v}$ its degree in $G$.   Note that $|\mathcal{P}|= \prod_{v\in V}(d_{v}+1)\leq 4^n$, since each vertex $v$ of $G$ can be either unmatched, or matched in exactly one of the gadgets $H^{(e_1)}, \hdots,H^{(e_{d_v})}$ where $e_1, \hdots, e_{d_v}$ are the edges incident to $v$. Fix a possible phase $Y\in \mathcal{P}$ and let $\Omega_Y$ be the set of matchings $M_R$ of $G_R$ such that $\mathcal{Y}(M_R)=Y$. The aggregate weight of matchings in $\Omega_Y$ is given by 
\[\prod_{e=(u,v)\in E}x^{(e)}_{R}(Y_{(e,u)}, Y_{(e,v)})=:W_Y.\]

Let $\mathcal{P}_1$ be the set of $Y\in \mathcal{P}$ such
that, for some edge   $\hat{e}=(\hat{u},\hat{v})\in E$,  we have $Y_{(\hat{e},\hat{u})}\neq Y_{(\hat{e},\hat{v})}$.
Let $\mathcal{P}_2 = \mathcal{P} \setminus \mathcal{P}_1$.
We next consider cases, depending on whether $Y\in \mathcal{P}_1$ or $Y\in \mathcal{P}_2$.
\begin{enumerate}
\item 
Suppose $Y \in \mathcal{P}_1$.
Let $\hat{e}=(\hat{u},\hat{v})$ 
be an edge in~$E$
such that $Y_{(\hat{e},\hat{u})}\neq Y_{(\hat{e},\hat{v})}$. 
In this case, using \eqref{eq:tv4tv444} and $|\gamma^{(e)}|-\epsilon\geq \epsilon$ and (hence) $|\gamma^{(e)}|+\epsilon\leq 2|\gamma^{(e)}|$, we obtain that 
\begin{align}
\frac{|W_Y|}{\prod_{e\in E}x^{(e)}_{R}(0,0)}&\leq \epsilon \prod_{e\in E; e\neq \hat{e}}(|\gamma^{(e)}|+\epsilon)\leq 2^m \epsilon \prod_{e\in E; e\neq 
\hat{e}}|\gamma^{(e)}|\notag\\
&\leq 2^m \epsilon |\gamma_0|^{m-2}\max\{|\gamma_0|, |R|\}\leq \epsilon'/5^n.\label{eq:ttvvr1}
\end{align}
\item 
Suppose $Y \in \mathcal{P}_2$.
Let $M$ be a subset of edges of $G$ such that $e\in M$ iff $Y_{(e,u)}= Y_{(e,v)}=1$. Since $Y\in \mathcal{P}_2$,   $M$ is a matching.  If $e^*\notin M$ then, using that $\frac{|a^n-b^n|}{|a-b|}=|\sum^{n-1}_{k=0}a^kb^{n-1-k}|\leq n\max\{|a|^{n-1},|b|^{n-1}\}$ for distinct real numbers $a,b$, we obtain using \eqref{eq:tv4tv444} that
\begin{align}
\Big|\frac{W_Y}{\prod_{e\in E}x^{(e)}_{R}(0,0)}-(\gamma_0)^{|M|}\Big|&=\Big|\prod_{e\in M}\frac{x^{(e)}_{R}(1,1)}{x^{(e)}_{R}(0,0)}-(\gamma_0)^{|M|}\Big|\notag\\
&\leq \max\Big\{\Big|(\gamma_0+\epsilon)^{|M|}-\gamma_0^{|M|}\Big|, \Big|(\gamma_0-\epsilon)^{|M|}-\gamma_0^{|M|}\Big|\Big\}\notag\\
&\leq \epsilon |M|(2|\gamma_0|)^{|M|-1}\leq \epsilon'/5^n.\label{eq:ttvvr2}
\end{align}
An analogous calculation shows that if $e^*\in M$ then with $A:=\prod_{e\in M; e\neq e^*}\frac{x^{(e)}_{R}(1,1)}{x^{(e)}_{R}(0,0)}$, we have that $|A-(\gamma_0)^{|M|-1}|\leq \epsilon |M|(2|\gamma_0|)^{|M|-2}$ and therefore
\begin{align}
\Bigg|\frac{W_Y}{\prod_{e\in E}x^{(e)}_{R}(0,0)}-(\gamma_0)^{|M|-1}R\Bigg|&=\Bigg|A\frac{x^{(e^*)}_{R}(1,1)}{x^{(e^*)}_{R}(0,0)}-(\gamma_0)^{|M|-1}R\Bigg|\notag\\
&=\Bigg|A\Big(\frac{x^{(e^*)}_{R}(1,1)}{x^{(e^*)}_{R}(0,0)}-R\Big)+R\Big(A-(\gamma_0)^{|M|-1}\Big)\Bigg|\notag\\
&\leq \epsilon(2|\gamma_0|)^{|M|-1}+ \epsilon |M| |R|(2|\gamma_0|)^{|M|-2}\leq \epsilon'/5^n.\label{eq:ttvvr3}
\end{align}
\end{enumerate}
Note that  
\[f_R=\frac{\sum_{Y\in \mathcal{P}} W_Y}{\prod_{e\in E}x^{(e)}_{R}(0,0)}\] and
\[\sum_{M\in \mathcal{M}_G; e^*\in M}(\gamma_0)^{|M|-1}=Z_{G\backslash\{u^*,v^*\}}(\gamma_0)=\alpha  \quad \text{and} \quad 
\sum_{M\in \mathcal{M}_G;e^*\notin M}(\gamma_0)^{|M|}=Z_{G-e^{*}}(\gamma_0)=\beta.
\]
So \eqref{eq:rtv4tv12}
is equivalent to
$$
\left|R 
 \sum_{M\in \mathcal{M}_G; e^*\in M}(\gamma_0)^{|M|-1}
+ 
\sum_{M\in \mathcal{M}_G;e^*\notin M}(\gamma_0)^{|M|}
- 
\frac{\sum_{Y\in \mathcal{P}_1} W_Y}{\prod_{e\in E}x^{(e)}_{R}(0,0)}
-
\frac{\sum_{Y\in \mathcal{P}_2} W_Y}{\prod_{e\in E}x^{(e)}_{R}(0,0)}
\right|\leq \epsilon'.
$$
Applying the triangle inequality and~\eqref{eq:ttvvr3},
the left-hand-side is at most
\[
 \sum_{Y\in \mathcal{P}_1}  \epsilon' / 5^n+\left|R 
 \sum_{M\in \mathcal{M}_G; e^*\in M}(\gamma_0)^{|M|-1}
+ 
\sum_{M\in \mathcal{M}_G;e^*\notin M}(\gamma_0)^{|M|}
-
\frac{\sum_{Y\in \mathcal{P}_2} W_Y}{\prod_{e\in E}x^{(e)}_{R}(0,0)}
\right|
.\]
Using once more  the triangle inequality, now in combination with   \eqref{eq:ttvvr2} and \eqref{eq:ttvvr3}, we obtain that the left-hand-side is at most  $\sum_{Y\in \mathcal{P}_1}  \epsilon' / 5^n+ \sum_{Y\in \mathcal{P}_2}  \epsilon' / 5^n\leq \epsilon$. This proves \eqref{eq:rtv4tv12}, therefore concluding the proofs of Theorems~\ref{thm:hard} and~\ref{thm:hardsign}.
\end{proof}

\section{Proof of Lemma~\ref{lem:expmatch} --- implementing edge activities}\label{sec:keylemmahard}
\subsection{Approximating the values of polynomials}
\begin{lemma}\label{lem:Polyapprox}
Let $P(x_1,\hdots, x_n)$ be a multivariate polynomial with rational coefficients.  Then, there is an algorithm which takes as input $P$ and rational numbers $\{a_i\}_{i=1,\hdots,n}$ and $\epsilon>0$, and outputs in time $poly(\size{P}, \size{a_1,\hdots,a_n, \epsilon})$ a rational number $\epsilon'>0$ such that for all real numbers $b_1,\hdots, b_n$ satisfying 
\[|b_1-a_1|,\hdots, |b_n-a_n|\leq \epsilon'\]
it holds that
\[\big|P(b_1,\hdots, b_n)- P(a_1,\hdots,a_n)\big|\leq \epsilon.\] 
\end{lemma}
\begin{proof}
Suppose that $P(x_1,\hdots,x_n)=\sum^m_{j=1}c_j \prod^{n}_{i=1}x_i^{d_{i,j}}$ where the $c_j$'s are non-zero rational numbers and the $d_{i,j}$'s are non-negative integers. In time $poly(\size{P}, \size{a_1,\hdots,a_n, \epsilon})$, we can compute a rational  $\epsilon'>0$ satisfying the following inequalities for all $i\in \{1,\hdots, m\}$ and $j\in \{1,\hdots, n\}$:
\begin{gather}
\epsilon'\leq |a_i|,\qquad \epsilon'd_{i,j} 2^{d_{i,j}-1}\leq |a_i|, \label{eq:3fff}\\
\epsilon'd_{i,j} (2|a_{i}|)^{d_{i,j}-1}\prod_{k<i} |a_{k}|^{d_{k,j}} \prod_{k>i} \big(2|a_{k}|^{d_{k,j}}\big)\leq \epsilon/(m n |c_j|).\label{eq:4fff}
\end{gather}
Now let $b_1,\hdots, b_n$ be arbitrary reals such that $|b_i-a_i|\leq \epsilon'$ for all $i=1,\hdots, n$, we will show that  $|P(b_1,\hdots, b_n)- P(a_1,\hdots,a_n)|\leq \epsilon$. 

We first show that for all $i\in \{1,\hdots, m\}$ and $j\in \{1,\hdots, n\}$, it holds that 
\begin{equation}\label{eq:5fff}
\big|a_i^{d_{i,j}}-b_i^{d_{i,j}}\big|\leq \epsilon'd_{i,j} (2|a_{i}|)^{d_{i,j}-1}.
\end{equation}
This is clear if $d_{i,j}=0$; if $d_{i,j}$ is a strictly positive integer, then we have
\begin{align*}
\big|a_i^{d_{i,j}}-b_i^{d_{i,j}}\big|&=|a_i-b_i|\cdot \Big|\sum^{d_{i,j}-1}_{k=0}a_i^{k}b_i^{d_{i,j}-1-k}\Big|\leq \epsilon' \sum^{d_{i,j}-1}_{k=0}|a_i|^{k}|b_i|^{d_{i,j}-1-k}\\
&\leq \epsilon' \sum^{d_{i,j}-1}_{k=0}|a_i|^{k}(|a_i|+\epsilon')^{d_{i,j}-1-k}\leq  \epsilon'\sum^{d_{i,j}-1}_{k=0}|a_i|^{k}(2|a_i|)^{d_{i,j}-1-k}\leq \epsilon'd_{i,j}(2|a_i|)^{d_{i,j}-1},
\end{align*}
where in the second-to-last inequality we used that $\epsilon'\leq |a_i|$ from \eqref{eq:3fff}. This proves \eqref{eq:5fff}. From \eqref{eq:3fff} we also have that $\epsilon'd_{i,j} (2|a_{i}|)^{d_{i,j}-1}\leq |a_i|^{d_{i,j}}$, so \eqref{eq:5fff} yields that
\begin{equation}\label{eq:5fffb}
|b_i|^{d_{i,j}}\leq 2|a_{i}|^{d_{i,j}}.
\end{equation}

Moreover, for arbitrary reals $\{y_i,z_i\}_{i=1,\hdots,n}$ we have the identity  
$\prod_{i}y_i-\prod_{j} z_i=\sum_i (y_i-z_i)\prod_{k<i} y_k \prod_{k>i} z_k$, which   gives that for all $j\in \{1,\hdots, n\}$ it holds that 
\begin{align*}
\Big|\prod^{n}_{i=1}a_i^{d_{i,j}}-\prod^{n}_{i=1}b_i^{d_{i,j}}\Big|&\leq \sum^n_{i=1}\big|a_i^{d_{i,j}}-b_i^{d_{i,j}}\big|\prod_{k<i} |a_{k}|^{d_{k,j}} \prod_{k>i} |b_k|^{d_{k,j}}\\
&\leq \epsilon'\sum^n_{i=1}d_{i,j} (2|a_{i}|)^{d_{i,j}-1}\big|\prod_{k<i} |a_{k}|^{d_{k,j}} \prod_{k>i} \big(2|a_{k}|^{d_{k,j}}\big)\leq \epsilon/(m  |c_j|),
\end{align*}
where the second-to-last inequality follows from \eqref{eq:5fff} and \eqref{eq:5fffb} and the last inequality follows from the choice of $\epsilon'$ in \eqref{eq:4fff}. It now remains to observe that 
\[\big|P(a_1,a_2,\hdots, a_n)-P(b_1,b_2,\hdots, b_n)\big|\leq \sum^m_{j=1}|c_j|\cdot\Big|\prod^{n}_{i=1}a_i^{d_{i,j}}-\prod^{n}_{i=1}b_i^{d_{i,j}}\Big|\leq \epsilon,\]
therefore completing the proof.
\end{proof}
We can extend Lemma~\ref{lem:Polyapprox} to rational functions.
\begin{lemma}\label{lem:Ratioapprox}
Let $P(x_1,\hdots, x_n), Q(x_1,\hdots,x_n)$ be multivariate polynomials with rational coefficients and let $f=P/Q$. 

There is an algorithm which  takes as input $P$ and $Q$, rational numbers $\{a_i\}_{i=1,\hdots,n}$ satisfying $Q(a_1,\hdots,a_n)\neq 0$ and a rational number $\epsilon>0$, and outputs in time $poly(\size{P,Q},\size{a_1,\hdots,a_n, \epsilon})$ a rational number $\epsilon'>0$ such that for all real numbers $b_1,\hdots, b_n$ satisfying 
\[|b_1-a_1|,\hdots, |b_n-a_n|\leq \epsilon'\]
it holds that $Q(b_1,\hdots,b_n)\neq 0$ and $\big|f(b_1,\hdots, b_n)- f(a_1,\hdots,a_n)\big|\leq \epsilon$.
\end{lemma}
\begin{proof}
Consider input polynomials $P$ and $Q$, and rational numbers  $a_1,\ldots,a_n$ and $\epsilon>0$ such that
$Q(a_1,\hdots,a_n)\neq 0$. 

Let $\eta:=\min\Big\{\frac{\epsilon|Q(a_1,\hdots,a_n)|^2}{2\big(|P(a_1,\hdots,a_n)|+|Q(a_1,\hdots,a_n)|\big)},\frac{1}{2}|Q(a_1,\hdots,a_n)|\Big\}$. Using the algorithm of Lemma~\ref{lem:Polyapprox}, we can compute in time $poly(\size{P,Q},\size{a_1,\hdots,a_n, \eta})=poly(\size{P,Q},\size{a_1,\hdots,a_n, \epsilon})$ a rational number $\epsilon'$ such that for all $b_1,\hdots, b_n$ satisfying $|b_1-a_1|,\hdots, |b_n-a_n|\leq \epsilon'$
it holds that 
\begin{equation}\label{eq:Qwwf}
\big|P(b_1,\hdots,b_n)-P(a_1,\hdots, a_n)\big|\leq \eta \quad\mbox{ and }\quad \big|Q(b_1,\hdots,b_n)-Q(a_1,\hdots, a_n)\big|\leq \eta.
\end{equation}
Let  $b_1,\hdots, b_n$ be arbitrary numbers satisfying  $|b_1-a_1|,\hdots, |b_n-a_n|\leq \epsilon'$, we will show that $Q(b_1,\hdots,b_n)\neq 0$ and  $|f(b_1,\hdots,b_n)-f(a_1,\hdots,a_n)|\leq \epsilon$. For convenience, let
\begin{equation*}
\begin{array}{lll}
N_1=P(a_1,\hdots,a_n),& & N_2=P(b_1,\hdots,b_n),\\
D_1=Q(a_1,\hdots,a_n),& & D_2=Q(b_1,\hdots,b_n).
\end{array}
\end{equation*}
Since $\eta\leq \frac{1}{2}|Q(a_1,\hdots,a_n)|$, we obtain from \eqref{eq:Qwwf} that
\begin{equation}\label{eq:ttv3t343fw4343}
|N_1-N_2|\leq \eta,\quad |D_1-D_2|\leq \eta,\quad |D_2|\geq \frac{1}{2}|D_1|.
\end{equation}
In particular, we have $D_2>0$. Moreover, we have that
\begin{align*}
\big|f(a_1,\hdots,a_n)-f(b_1,\hdots,b_n)\big|&= \frac{|N_1D_2-N_2D_1|}{|D_1D_2|}=\frac{|N_1(D_2-D_1)+D_1(N_1-N_2)|}{|D_1D_2|}\\
&\leq\frac{|N_1||D_2-D_1|+|D_1||N_2-N_1|}{|D_1D_2|}\\
&\leq \eta \frac{|N_1|+|D_1|}{|D_1||D_2|}\leq \eta \frac{2(|N_1|+|D_1|)}{|D_1|^2}\leq \epsilon,
\end{align*}
where in the last three inequalities we used \eqref{eq:ttv3t343fw4343} and  the choice of $\eta$. This concludes the proof.
\end{proof}

\subsection{Implementing vertex activities}
\begin{definition}\label{def:verteximpl}
Fix a real number $\gamma$. Given $\gamma$,
the graph $G=(V,E)$  is said to \emph{implement the vertex activity} $\lambda\in \Reals$ with \emph{accuracy} $\epsilon>0$ if there is vertex $u$ in $G$ such that  
\begin{enumerate}
\item \label{it:uone} $u $ has degree one in $G$, 
\item  \label{it:123ps1}  $Z_{G}(\gamma)\neq 0$ and  $\displaystyle \Big|\frac{Z_{G,\neg u}(\gamma)}{Z_{G}(\gamma)}-\lambda\Big|\leq \epsilon $.
\end{enumerate}
We call $u$ the terminal of $G$. If Item~\ref{it:123ps1} holds with $\epsilon=0$, we say that $G$ implements $\lambda$ (perfectly).
\end{definition}

\newcommand{\statelemexpvertmatch}{Let $\Delta\geq 3$ be an integer and $\gamma<-\frac{1}{4(\Delta-1)}$ be a rational number.

There is an algorithm which, on input a rational number $\lambda$ and $\eps>0$, outputs in $poly(\size{\lambda,\eps})$ time a bipartite graph $G$ of maximum degree at most  $\Delta$ that implements the vertex activity $\lambda$ with accuracy $\eps$.}
\begin{lemma}\label{lem:expvertmatch}
\statelemexpvertmatch
\end{lemma}

Assuming Lemma~\ref{lem:expvertmatch} for now, we can conclude Lemma~\ref{lem:expmatch}.
\begin{proof}[Proof of Lemma~\ref{lem:expmatch}]
Let $\gamma'\leq  0$ be an arbitrary rational edge activity that we wish to implement with some accuracy $\epsilon>0$. We can compute in time $poly(\size{\gamma',\epsilon})$ a rational $\gamma''<0$ such that $|\gamma'+(\gamma'')^{2}|\leq \epsilon/2$. Thus, to implement the edge activity $\gamma'$ with accuracy $\epsilon$, it suffices to implement the edge activity $-(\gamma'')^{2}$ with accuracy $\epsilon/2$ (this way we will avoid irrational square roots in the following argument). We begin by specifying some parameters that will be important later. 

Let $\lambda_1=-\gamma''/\gamma, \lambda_2=1/\gamma'', \lambda_3=-\gamma''/\gamma$ and note that 
\[1+\gamma\lambda_1 \lambda_2=1+\gamma\lambda_2 \lambda_3=0,\quad  \frac{\gamma^2\lambda_1\lambda_3}{1+\gamma\lambda_1\lambda_2+\gamma \lambda_2\lambda_3}=-(\gamma'')^2.\]
Consider the multivariate polynomials 
\begin{gather*}
P_1(x_1,x_2,x_3)=\gamma x_1(1+\gamma x_2 x_3),\quad P_2(x_1,x_2,x_3)=\gamma^2x_1x_3,\quad P_3(x_1,x_2,x_3)=\gamma x_3(1+\gamma x_1 x_2)\\
Q(x_1,x_2,x_3)=1+\gamma x_1 x_2+\gamma x_2 x_3.
\end{gather*}
For $i\in \{1,2,3\}$, let $f_i=P_i/Q$ so that 
\[f_1(\lambda_1,\lambda_2,\lambda_3)=f_3(\lambda_1,\lambda_2,\lambda_3)=0, \quad f_2(\lambda_1,\lambda_2,\lambda_3)=-(\gamma'')^2.\]
Using the algorithm of Lemma~\ref{lem:Ratioapprox}, we can compute in $poly(\size{P_1,P_2,P_3,Q}, \size{\lambda_1,\lambda_2,\lambda_3,\epsilon})=poly(\size{\gamma',\epsilon})$ time a rational number $\epsilon'$ such that for all $\lambda_1',\lambda_2',\lambda_3'$ satisfying $|\lambda_1-\lambda_1'|,|\lambda_2-\lambda_2'|, |\lambda_3-\lambda_3'|\leq \epsilon'$ it holds that 
\begin{equation}\label{eq:rc3r664c3}
Q(\lambda_1',\lambda_2',\lambda_3')\neq 0 \mbox{\ \ and \ \ }\big|f_i(\lambda_1',\lambda_2',\lambda_3')-f_i(\lambda_1,\lambda_2,\lambda_3)\big|\leq \epsilon/2 \mbox{\ \ for all $i\in \{1,2,3\}$.}
\end{equation}

Let $i\in\{1,2,3\}$. Using the algorithm of Lemma~\ref{lem:expvertmatch},  we can construct in $poly(\size{\lambda_i,\epsilon'})=poly(\size{\gamma',\epsilon})$ time  a bipartite graph $H_i$ of maximum degree at most $\Delta$ that implements the vertex activity $\lambda_i$ with accuracy $\epsilon'$. Let $y_i$ be the terminal of $H_i$ and let 
\begin{equation}\label{eq:43frf4f3}
q_i:=Z_{H_i,\neg y_i}(\gamma),\quad z_i:=Z_{H_i}(\gamma)\quad \mbox{ so that } z_i\neq 0 \mbox{ and } \Big|\frac{q_i}{z_i}-\lambda_i\Big|\leq \epsilon'.
\end{equation}
Let $G$ be the graph obtained by taking the disjoint union of $H_1,H_2,H_3$, two new vertices $u,v$ and adding the edges $(u,y_1), (y_1,y_2), (y_2,y_3),(y_3,v)$. Note that $G$ is bipartite and $u,v$ lie in the same part of the vertex partition of $G$. Then, we have that
\begin{equation}\label{eq:ZGuvall}
\begin{array}{llll}
Z_{G,u,v}(\gamma)=\gamma^2 q_1 z_2 q_3,&  & & Z_{G, \neg u, \neg v}(\gamma)=z_1 z_2 z_3 +\gamma (q_1 q_2 z_3+z_1 q_2 q_3),\\
Z_{G,  u,\neg v}(\gamma)=\gamma q_1 (z_2 z_3+\gamma q_2 q_3),& & & Z_{G, \neg u,v}(\gamma)=\gamma q_3 (z_1 z_2+\gamma q_1 q_2).
\end{array}
\end{equation}
For illustration, we next justify the expression for $Z_{G,u,\neg v}(\gamma)$. Let $M$ be a matching such that $u$ is matched but not $v$.  Note that the only way  that $u$ can be matched in $G$ is if $(u,y_1)\in M$. Since $v$ is unmatched in $M$ and $(u,y_1)\in M$, the edges $(y_1,y_2)$ and $(y_3,v)$ do not belong to $M$. The remaining edge $(y_2,y_3)$ can either belong to $M$ or not. The aggregate contribution of matchings $M$ with $(y_2,y_3)\in M$ is $\gamma^2 q_1 q_2 q_3$: the factor $\gamma^2$ comes from the edges $(u,y_1)$ and $(y_2,y_3)$  and the factor $q_1q_2q_3$ comes from the fact that $y_1,y_2,y_3$ are unmatched in $H_1,H_2,H_3$ respectively. The aggregate contribution of matchings $M$ with $(y_2,y_3)\notin M$ is $\gamma q_1 z_2 z_3$: the factor $\gamma$ comes from the edge $(u,y_1)$, the factor $q_1$ comes from the fact that $y_1$ is unmatched in $H_1$ and the factor $z_2z_3$ accounts for the weight of all matchings in $H_2, H_3$ (note that there is no restriction in this case on the vertices $y_2,y_3$ whether they are matched). This proves the expression for $Z_{G,u,\neg v}(\gamma)$ in \eqref{eq:ZGuvall}, the remaining expressions therein can be justified analogously.

Let $\lambda_1'=q_1/z_1, \lambda_2'=q_2/z_2, \lambda_3'=q_3/z_3$. Then, using \eqref{eq:ZGuvall}, we see that
\begin{equation}\label{eq:tvt4vr}
Z_{G, \neg u, \neg v}(\gamma)=z_1z_2z_3\Big(1+\gamma\frac{q_1q_2}{z_1z_2}+\gamma\frac{q_2q_3}{z_2z_3}\Big)=z_1z_2z_3\cdot  Q(\lambda_1',\lambda_2',\lambda_3')\neq 0,
\end{equation}
where the disequality follows from $z_1z_2z_3\neq 0$ (cf. \eqref{eq:43frf4f3}) and $Q(\lambda_1',\lambda_2',\lambda_3')\neq 0$ from \eqref{eq:rc3r664c3}; note that the $\lambda_j'$ fall within the scope of \eqref{eq:rc3r664c3} since \eqref{eq:43frf4f3} guarantees that $|\lambda_1-\lambda_1'|,|\lambda_2-\lambda_2'|, |\lambda_3-\lambda_3'|\leq \epsilon'$. Using again \eqref{eq:ZGuvall}, we  have that
\begin{equation}\label{eq:tvt4vrv}
\Big|\frac{Z_{G,u,v}(\gamma)}{Z_{G, \neg u, \neg v}(\gamma)}+(\gamma'')^2\Big|=\Big|\frac{\gamma^2\frac{q_1 q_3}{z_1 z_3}}{1+\gamma (\frac{q_1 q_2}{z_1 z_2}+\frac{q_2q_3}{z_2z_3})}+(\gamma'')^2\Big|=\big|f_2(\lambda_1',\lambda_2',\lambda_3')-f_2(\lambda_1,\lambda_2,\lambda_3)\big|\leq \epsilon/2,
\end{equation}
where the last inequality follows from \eqref{eq:rc3r664c3} for $i=2$. Analogously, we obtain that
\begin{equation}\label{eq:tvt4vrve}
\begin{aligned}
\Big|\frac{Z_{G,u,\neg v}(\gamma)}{Z_{G, \neg u, \neg v}(\gamma)}\Big|&=\Big|\frac{\gamma \frac{q_1}{z_1}\big(1+\gamma\frac{q_2 q_3}{z_2z_3}\big)}{1+\gamma \big(\frac{q_1 q_2}{z_1 z_2}+\frac{q_2q_3}{z_2z_3}\big)}\Big|=\big|f_1(\lambda_1',\lambda_2',\lambda_3')-f_1(\lambda_1,\lambda_2,\lambda_3)\big|\leq \epsilon/2,\\
\Big|\frac{Z_{G,\neg u,v}(\gamma)}{Z_{G, \neg u, \neg v}(\gamma)}\Big|&=\Big|\frac{\gamma \frac{q_3}{z_3}\big(1+\gamma\frac{q_1 q_2}{z_1z_2}\big)}{1+\gamma \big(\frac{q_1 q_2}{z_1 z_2}+\frac{q_2q_3}{z_2z_3}\big)}\Big|=\big|f_3(\lambda_1',\lambda_2',\lambda_3')-f_3(\lambda_1,\lambda_2,\lambda_3)\big|\leq \epsilon/2.
\end{aligned}
\end{equation}
Combining \eqref{eq:tvt4vr}, \eqref{eq:tvt4vrv} and \eqref{eq:tvt4vrve}, we obtain that the bipartite graph $G$ (with terminals $u,v$) implements the edge activity $-(\gamma'')^2$ with accuracy $\epsilon/2$. Since $|\gamma'+(\gamma'')^2|\leq \epsilon/2$, we therefore obtain that $G$ implements  the edge activity $\gamma'$ with accuracy $\epsilon$, as wanted. This concludes the proof of Lemma~\ref{lem:expmatch}.
\end{proof}

To prove Lemma~\ref{lem:expvertmatch}, we will first  need to prove the following lemma.

\newcommand{\statelemvertmatch}{Let $\Delta\geq 3$ be an integer and $\gamma<-\frac{1}{4(\Delta-1)}$ be a real number.

For every $\lambda\in \mathbb{R}$ and $\eps>0$, there is a bipartite graph $G$ of maximum degree at most  $\Delta$ that implements the vertex activity $\lambda$ with accuracy $\eps$.}
\begin{lemma}\label{lem:vertmatch}
\statelemvertmatch
\end{lemma}

\subsection{Proof of Lemma~\ref{lem:vertmatch}}
To prove Lemma~\ref{lem:vertmatch}, we will need to consider two cases for the value of $\gamma$. Namely, for an integer $\Delta\geq 3$,  the following subset of the negative reals will be relevant:
\begin{equation}\label{def:bdelta}
\mathcal{B}_{\Delta}=\Big\{\gamma\in \mathbb{R}\mid \gamma=-\frac{1}{4(\Delta-1)(\cos \theta)^2} \mbox{ for some $\theta\in (0,\pi/2) $ which is a rational multiple of $\pi$}\Big\}.
\end{equation}

\begin{lemma}\label{lem:notbad}
Let $\Delta\geq 3$ be an integer and $\gamma<-\frac{1}{4(\Delta-1)}$. For an integer $n\geq 0$, let $T_n$ be the $(\Delta-1)$-ary tree of height $n$ rooted at $\rho_n$.
\begin{enumerate}
\item \label{it:33rc3d22d} If $\gamma\notin \mathcal{B}_{\Delta}$, then $Z_{T_n}(\gamma)\neq 0$ for all $n\geq 0$. Moreover, for every $\lambda \in \mathbb{R}$ and $\epsilon>0$, there exists $n$ such that
\[\Big|\frac{Z_{T_n,\neg \rho_n}(\gamma)}{Z_{T_n}(\gamma)}-\lambda\Big|\leq \epsilon.\]
\item \label{it:33rc3d22e} If $\gamma\in \mathcal{B}_{\Delta}$, there exists $n$ such that  $Z_{T_n}(\gamma)=0$. Moreover, there exists a tree of maximum degree at most $\Delta$ which implements either the edge activity $\gamma'=-1$ or the edge activity $\gamma'=-1/4$ (perfectly).
\end{enumerate}
\end{lemma}
\begin{proof}
Fix arbitrary $\gamma<-\frac{1}{4(\Delta-1)}$ and let $\theta\in (0,\pi/2)$ be such that $\gamma=-\frac{1}{4(\Delta-1)(\cos \theta)^2}$. For integer $n\geq 0$, set 
\[u_n=Z_{T_n, \neg \rho_n}(\gamma), \quad z_n=Z_{T_n}(\gamma).\] 
Note that $u_0=z_0=1$, while for $n\geq 1$ we have that
\begin{equation}\label{eq:4r34f3d}
u_n=\big(z_{n-1})^{\Delta-1}, \quad z_n=\big(z_{n-1})^{\Delta-1}+(\Delta-1)\gamma\big(z_{n-1}\big)^{\Delta-2}u_{n-1}.
\end{equation}
Moreover, let $W_n:=\displaystyle \frac{\sin ((n+1) \theta)}{\sin ((n+2)\theta)}$. \vskip 0.3cm

\noindent \textbf{Proof of Item~\ref{it:33rc3d22d}: $\gamma\notin \mathcal{B}_\Delta$, i.e., $\theta$ is not a rational multiple of $\pi$.} In this case, we clearly have that $W_n\neq 0$ for all $n\geq 0$. We will show by induction on $n$ that
\begin{equation}\label{eq:dcf2422}
z_n\neq 0 \quad \mbox{and} \quad \frac{u_n}{z_n}=(2\cos \theta) W_n.
\end{equation}
Indeed, \eqref{eq:dcf2422} holds for $n=0$ (using the identity $\sin(2\theta)=2\sin \theta \cos \theta$). Now, let $n\geq 1$ and suppose that \eqref{eq:dcf2422} holds for all integers less than $n$. We first show that $z_n\neq 0$. From \eqref{eq:4r34f3d}, we have that
\begin{align}
z_n&=(z_{n-1})^{\Delta-1}\Big(1+(\Delta-1)\gamma\frac{u_{n-1}}{z_{n-1}}\Big)\notag\\
&=(z_{n-1})^{\Delta-1}\Big(1-\frac{1}{2\cos \theta}W_{n-1}\Big).\label{eq:rv3tvt3}
\end{align}
Dividing the identity  $2\cos \theta\sin((n+1)\theta)-\sin (n \theta)=\sin ((n+2) \theta)$ by $2\cos \theta \sin((n+1)\theta)$ (note that this is non-zero by our assumption on  $\theta$), we obtain that
\begin{equation*}
1-\frac{1}{2\cos \theta}W_{n-1}=\frac{1}{(2\cos \theta)W_n}.
\end{equation*}
From this and \eqref{eq:rv3tvt3}, we obtain that $z_n=\frac{(z_{n-1})^{\Delta-1}}{(2\cos \theta)W_n}$. We therefore have that $z_n\neq 0$. Moreover, \eqref{eq:4r34f3d} yields that $u_n/z_n=(2\cos \theta)W_n$, therefore completing the proof of \eqref{eq:dcf2422}.

In light of \eqref{eq:dcf2422}, we immediately have that $Z_{T_n}(\gamma)\neq 0$ for all integers $n\geq 0$, therefore proving the first part of Item~\ref{it:33rc3d22d}. For the second part,  \cite[Proof of Lemma 11]{GGS} shows that the sequence $W_n$ is dense in $\mathbb{R}$ as $n$ ranges over the positive integers. Since we have that  $u_n/z_n=(2\cos \theta)W_n$ for all $n\geq 0$ from \eqref{eq:dcf2422}, this immediately yields the second part of Item~\ref{it:33rc3d22d}.
\vskip 0.3cm

\noindent  \textbf{Proof of Item~\ref{it:33rc3d22e}: $\gamma\in \mathcal{B}_\Delta$, i.e., $\theta$ is a rational multiple of $\pi$.}  Since $\theta\in (0,\pi/2)$, we have that $\sin(\theta), \sin(2\theta)\neq 0$. Let $n_0$ be the smallest non-negative integer such that $\sin((n+2) \theta)=0$. We first show that $Z_{T_{n_0}}(\gamma)= 0$.

Note first that $W_n$ is well-defined and non-zero for all non-negative integers $<n_0$.  As in Case I above, we obtain by induction on $n$ that for all $n=0,1\hdots, n_0-1$, it holds  that
\begin{equation}\label{eq:dcf2422b}
z_n\neq 0 \quad \mbox{and} \quad \frac{u_n}{z_n}=(2\cos \theta) W_n.
\end{equation}
Just as in \eqref{eq:rv3tvt3}, we then obtain that
\[z_{n_0}=\big(z_{n_0-1}\big)^{\Delta-1}\Big(1-\frac{1}{2\cos \theta}W_{n_0-1}\Big).\]
Now, using again the identity  $2\cos \theta\sin((n+1)\theta)-\sin (n \theta)=\sin ((n+2) \theta)$ for $n=n_0$, we obtain that $2\cos \theta\sin((n_0+1)\theta)-\sin (n_0 \theta)=0$ and dividing by $2\cos \theta \sin((n_0+1)\theta)$, we obtain that $1-\frac{1}{2\cos \theta}W_{n_0-1}=0$. It follows that $z_{n_0}=0$, therefore proving that $Z_{T_{n_0}}(\gamma)= 0$, as wanted.

To complete the proof of Item~\ref{it:33rc3d22e}, it remains to show that there exists a tree of maximum degree at most $\Delta$ which implements either the edge activity $\gamma'=-1$ or the edge activity $\gamma'=-1/4$ perfectly.
 For $\gamma=-1$,  the result follows by considering the path with four vertices (and using the endpoints as terminals), so we assume $\gamma\neq -1$ in what follows. 

To implement the desired edge activities for $\gamma\neq -1$, we will adapt the construction in the proof of Lemma~\ref{lem:expmatch}. Namely, for $i\in\{1,2,3\}$, let $H_i$ be  a tree of maximum degree $\Delta$ such that $Z_{H_i}(\gamma)\neq 0$ and suppose that $y_i$ is a vertex in $H_i$ of degree $\leq \Delta-2$. Let
\begin{equation}\label{eq:43frf4f3b}
q_i:=Z_{H_i,\neg y_i}(\gamma),\quad z_i:=Z_{H_i}(\gamma)\neq 0.
\end{equation}
Let $G$ be the tree obtained by taking the disjoint union of $H_1,H_2,H_3$, two new vertices $u,v$ and adding the edges $(u,y_1), (y_1,y_2), (y_2,y_3),(y_3,v)$. Note that the restriction of the degrees of $y_i$'s in the graph $H_i$'s ensures that $G$ is a tree of maximum degree at most $\Delta$. Then, we have that
\begin{equation*}\tag{\ref{eq:ZGuvall}}
\begin{array}{llll}
Z_{G,u,v}(\gamma)=\gamma^2 q_1 z_2 q_3,&  & & Z_{G, \neg u, \neg v}(\gamma)=z_1 z_2 z_3 +\gamma (q_1 q_2 z_3+z_1 q_2 q_3),\\
Z_{G,  u,\neg v}(\gamma)=\gamma q_1 (z_2 z_3+\gamma q_2 q_3),& & & Z_{G, \neg u,v}(\gamma)=\gamma q_3 (z_1 z_2+\gamma q_1 q_2).
\end{array}
\end{equation*} Among those trees $T'$ of maximum degree $\Delta$ such that $Z_{T'}(\gamma)= 0$, let $T$ be a tree with the minimum number of vertices. Since $\gamma\neq -1$, we have that $T$ has more than two vertices. 

In the tree $T$ (more generally, in any tree with more than two vertices), either there exists a vertex $p$ with at least two leaves as children or else there exists a leaf $l$ whose parent $p$ has degree one in the tree $T\backslash \{l\}$.  We consider the two cases separately.
\begin{enumerate}
\item There exists a leaf $l$ whose parent $p$ has degree one in the tree $T\backslash \{l\}$.  Let $T^*=T\backslash \{l\}$. Since $Z_{T}(\gamma)=0$, we have that $Z_{T,l}(\gamma)+Z_{T,\neg l}(\gamma)=0$. We have $Z_{T,l}(\gamma)=\gamma Z_{T^*,\neg p}(\gamma)$ and $Z_{T,\neg l}(\gamma)=Z_{T^*}(\gamma)$. Moreover, we have that $Z_{T^*}(\gamma)\neq 0$ from the choice of $T$. Hence
\begin{equation}\label{eq:4f4rf4f4b}
\gamma Z_{T^*,\neg p}(\gamma)+Z_{T^*}(\gamma)=0, \quad Z_{T^*}(\gamma)\neq 0.
\end{equation}
Now we let $H_1,H_3$ be disjoint copies of $T^*$ and set $y_1$, $y_3$ to be the corresponding copies of $p$. Let $H_2$ be the single-vertex graph consisting only of the vertex $y_2$. Then, we have that 
\[q_1=q_3=Z_{T^*,\neg p}(\gamma),\mbox{\ \ }z_1=z_3=Z_{T^*}(\gamma), \mbox{ and } q_2=z_2=1,\]
so \eqref{eq:ZGuvall} yields that 
\[Z_{G,  u,\neg v}(\gamma)=Z_{G,  \neg u,v}(\gamma)=0 \mbox{ and } \frac{Z_{G,  u,v}(\gamma)}{Z_{G, \neg u,\neg v}(\gamma)}=-1.\]
Therefore, $G$ with terminals $u$ and $v$ implements the edge activity $\gamma'=-1$. 

\item There exists a vertex $p$ with at least two leaves as children, say $l$ and $l'$.  Note that $l$ and $l'$ cannot be simultaneously matched in $T$ and therefore, since $Z_{T}(\gamma)=0$, we have that 
\[Z_{T,l,\neg l'}(\gamma)+Z_{T,\neg l, l'}(\gamma)+Z_{T,\neg l, \neg l'}(\gamma)=0.\]
Let $T^*=T\backslash \{l,l'\}$ and note that $p$ has degree $\leq \Delta-2$ in $T^*$.  We have 
\[\mbox{$Z_{T,l,\neg l'}(\gamma)=Z_{T,\neg l, l'}(\gamma)=\gamma Z_{T^*,\neg p}(\gamma)$ and $Z_{T,\neg l,\neg l'}(\gamma)=Z_{T^*}(\gamma)$.}\] Moreover, we have that $Z_{T^*}(\gamma)\neq 0$ from the choice of $T$. Hence,
\begin{equation}\label{eq:yrtyrf4f4c}
2\gamma Z_{T^*,\neg p}(\gamma)+Z_{T^*}(\gamma)=0, \quad Z_{T^*}(\gamma)\neq 0.
\end{equation}
Now we let $H_1,H_3$ be disjoint copies of $T^*$ and set $y_1$, $y_3$ to be the corresponding copies of $p$. Then, from \eqref{eq:yrtyrf4f4c},  we have that 
\[2\gamma q_1+z_1=2\gamma q_3+z_3=0, \quad z_1,z_3\neq 0.\]
Let $H_2$ be the tree obtained from  $T^*$ by adding the vertex $y_2$ and connecting it to $p$. Then, we have that 
\[q_2=Z_{H_2,\neg y_2}(\gamma)= Z_{T^*}(\gamma) \mbox{ and } z_2=Z_{H_2}(\gamma)= Z_{T^*}(\gamma)+\gamma Z_{T^*,\neg p}(\gamma),\]
so from \eqref{eq:yrtyrf4f4c} we obtain that 
\[-\frac{1}{2}q_2+z_2=0, \quad z_2\neq 0.\]
Equation \eqref{eq:ZGuvall} then gives  that 
\[Z_{G,  u,\neg v}(\gamma)=Z_{G,  \neg u,v}(\gamma)=0 \mbox{ and } \frac{Z_{G,  u,v}(\gamma)}{Z_{G, \neg u,\neg v}(\gamma)}=-1/4,\]
and, therefore, $G$ with terminals $u$ and $v$ implements the edge activity $\gamma'=-1/4$.
\end{enumerate}
This completes the proof of Item~\ref{it:33rc3d22e} in the lemma, and therefore establishes Lemma~\ref{lem:notbad}. 
\end{proof}

\subsection{Handling the edge activities $-1$ and $-1/4$}
The proof of the following lemma builds upon the techniques  in \cite[Lemmas 21 \& 22]{GGS}.
\begin{lemma}\label{lem:gammaminusone}
Let $\gamma=-1$. Then, for every rational number $\lambda$, there exists 
a tree of maximum degree $\Delta=3$ that implements the vertex activity $\lambda$. 
\end{lemma}
\begin{proof}
Since the value of $\gamma$ is fixed to $-1$, for this proof we will omit it from the notation of partition functions, i.e., for a  graph $G$ we will simply write $Z_G$ instead of $Z_G(\gamma)$.

Consider the functions 
\[f_1(x)=\frac{1}{1-x} \mbox{ for $x\neq 1$}\quad \mbox{ and }\quad f_2(x)=-\frac{1}{x} \mbox{ for $x\neq 0$}\]
Let $S$ be the set of all numbers $s$ for which there exist an integer $n\geq 0$ and indices $i_1,\hdots,i_n\in\{1,2\}$ such that the sequence given by
\begin{equation}\label{eq:defS}
x_0=-1, \quad x_{j}=f_{i_j}(x_{j-1}) \mbox{ for }j\in\{1,\hdots,n\}
\end{equation}
satisfies $x_n=s$ (note that $-1\in S$ by taking $n=0$).  We will show the following.
\begin{enumerate}
\item \label{it:Sim} Let $i_1,\hdots,i_n\in \{1,2\}$ be arbitrary indices and consider the sequence $x_0,x_1,\hdots, x_n$ defined in \eqref{eq:defS}. Then, for every $j\in \{0,1,\hdots, n\}$, there exists a tree $G_j$ of maximum degree at most $3$ and a vertex $u_j$ in $G_j$ which has degree at most 2 such that $\frac{Z_{G_j,\neg u_j}}{Z_{G_j}}=x_j$. 
\item \label{it:SQ} $S=\mathbb{Q}\backslash\{0\}$, i.e., $S$ is the set of all non-zero rational numbers.
\end{enumerate} 
To conclude the lemma from these two Items, let $\lambda$ be an arbitrary rational number. If $\lambda=0$, then we can implement the vertex activity $\lambda$ with a three-vertex path and using as terminal one of its endpoints. If $\lambda=1$, then we can implement the vertex activity $\lambda$ with a four-vertex path and using as  terminal one of its endpoints. For $\lambda\neq 0,1$, let $s=\frac{\lambda-1}{\lambda}$ and note that $s\in \mathbb{Q}\backslash \{0\}$. Then, by Item~\ref{it:SQ}, there exists a sequence $x_1,\hdots,x_n$ of the form \eqref{eq:defS} such that $x_n=s$. Therefore, by Item~\ref{it:Sim}, there exists a tree $G$ of maximum degree at most $3$ and a vertex $u$ in $G$ which has degree at most 2 such that $\frac{Z_{G,\neg u}}{Z_G}=s$. Add a new vertex $u'$ to $G$ and connect it to $u$, and let $G'$ be the resulting graph. Then, we have that
\begin{equation}\label{eq:5vtvtvtqa}
Z_{G',\neg u'}=Z_G \mbox{ and } Z_{G'}=Z_{G',\neg u'}+Z_{G',u'}=Z_G+\gamma Z_{G,\neg u}=Z_G-Z_{G,\neg u}.
\end{equation}
It follows that $G'$ with terminal $u'$ is a tree of maximum degree $3$ that implements the vertex activity $\frac{Z_{G',\neg u'}}{Z_{G'}}=\frac{Z_G}{Z_G-Z_{G,\neg u}}=\frac{1}{1-s}=\lambda$. 

To prove Item~\ref{it:Sim}, we use induction on $j$. For the base case $j=0$, we can take the tree $G_0$ to be the path with three vertices and $u_0$ to be the vertex in the middle of the path. So assume that $j\geq 1$. To obtain $G_{j}$ from $G_{j-1}$, we consider cases on the value of $i_{j}$. 
\begin{enumerate}
\item If $i_{j}=1$, we have $x_{j}=f_1(x_{j-1})$ and therefore $x_{j-1}\neq 1$. Let $G_{j}$ be the tree obtained from $G_{j-1}$ by adding a new vertex $u_{j}$ and joining it to $u_{j-1}$. Then, analogously to \eqref{eq:5vtvtvtqa}, we obtain that 
\begin{equation*}
Z_{G_{j},\neg u_{j}}=Z_{G_{j-1}} \mbox{ and } Z_{G_{j}}=Z_{G_{j-1}}-Z_{G_{j-1},\neg u_{j-1}}\neq 0,
\end{equation*}
where the disequality follows from $x_{j-1}=\frac{Z_{G_{j-1},\neg u_{j-1}}}{Z_{G_{j-1}}}\neq 1$. We therefore obtain that $\frac{Z_{G_{j},\neg u_{j}}}{Z_{G_{j}}}=\frac{Z_{G_{j-1}}}{Z_{G_{j-1}}-Z_{G_{j-1},\neg u_{j-1}}}=\frac{1}{1-x_{j-1}}=f_1(x_{j-1})=x_{j}$, as wanted.
\item If $i_{j}=2$, we have $x_{j}=f_2(x_{j-1})$ and therefore $x_{j-1}\neq 0$. Let $G_{j}$ be the tree obtained from $G_{j-1}$ by adding two new vertices $u_{j},w$ and joining $u_{j}$ to both $u_{j-1}$ and $w$ (note that $w$'s only neighbour is $u_{j}$). Then, analogously to \eqref{eq:5vtvtvtqa}, we obtain that 
\begin{equation*}
Z_{G_{j},\neg u_{j}}=Z_{G_{j-1}} \mbox{ and } Z_{G_{j}}=Z_{G_{j-1}}+\gamma(Z_{G_{j-1}}+Z_{G_{j-1},\neg u_{j-1}})=-Z_{G_{j-1},\neg u_{j-1}}\neq 0,
\end{equation*}
where the disequality follows from $x_{j-1}=\frac{Z_{G_{j-1},\neg u_{j-1}}}{Z_{G_{j-1}}}\neq 0$. We therefore obtain that $\frac{Z_{G_{j},\neg u_{j}}}{Z_{G_{j}}}=-\frac{Z_{G_{j-1}}}{Z_{G_{j-1},\neg u_{j-1}}}=-\frac{1}{x_{j-1}}=f_2(x_{j-1})=x_{j}$, as wanted.
\end{enumerate}
This finishes the proof of Item~\ref{it:Sim}.

To prove Item~\ref{it:SQ}, first note that the inclusion $S\subseteq \mathbb{Q}\backslash\{0\}$ is trivial, so it suffices to show that $\mathbb{Q}\backslash\{0\}\subseteq S$. The following equalities will be useful:
\begin{align}
f_1(f_1(f_2(x)))&=1+x\mbox{  for all $x\neq 0,-1$,}\label{eq:er123qwe12a}\\
f_2(f_1(f_1(x)))&=\frac{x}{1-x} \mbox{ for all $x\neq 0,1$},\label{eq:er123qwe12b}\\
f_2(f_1(f_1(f_2(f_1(f_1(f_2(x)))))))&=-\frac{1}{2+x}\mbox{  for all $x\neq 0,-1,-2$.}\label{eq:er123qwe12c}
\end{align}
Establishing \eqref{eq:er123qwe12a} and \eqref{eq:er123qwe12b} is straightforward, and \eqref{eq:er123qwe12c} follows by applying \eqref{eq:er123qwe12a} after writing the left-hand-side as $f_2(g(g(x)))$ where $g(x)=f_1(f_1(f_2(x)))$.

We next show that  
\begin{equation}\label{eq:wsx12}
(-1,0)\cap \mathbb{Q}\subseteq S.
\end{equation}
For the sake of contradiction, assume that there exists a rational number $\alpha\in (-1,0)$ such that $\alpha\notin S$. Note that $\alpha=-p/q$ for some relative coprime positive integers $p,q$ satisfying $p<q$; we may assume that $\alpha$ is such that the sum $p+q$ is minimised over all rational $\alpha\in (-1,0)$ with $\alpha\notin S$. We consider three cases:
\begin{enumerate}
\item $q=2p$. Since $p,q$ are coprime positive integers, this means that $p=1$ and $q=2$ and therefore $\alpha=-1/2$. But $-1/2\in S$ since $f_2(f_1(f_1(-1)))=-1/2$.
\item $q>2p$. Let $p'=p$, $q'=q-p$  and note that $p',q'$ are coprime positive integers satisfying $p'<q'$. Since $p'+q'<p+q$ we conclude that the rational number $\alpha'=-p'/q'$ belongs to $S$. But from \eqref{eq:er123qwe12b} we have that  $f_2(f_1(f_1(\alpha')))=\frac{p'/q'}{1+p'/q'}=-p/q=\alpha$, contradicting that $\alpha\notin S$. 
\item $q<2p$. Let $p'=2p-q$, $q':=p$ and note that $p',q'$ are coprime positive integers satisfying $p'<q'$. Since $p'+q'<p+q$ we conclude that the rational number $\alpha'=-p'/q'$ belongs to $S$. But from \eqref{eq:er123qwe12b} we have that  $f_2(f_1(f_1(f_2(f_1(f_1(f_2(\alpha')))))))=-\frac{1}{2-p'/q'}=-p/q=\alpha$, contradicting that $\alpha\notin S$. 
\end{enumerate}
This finishes the proof of \eqref{eq:wsx12}. From \eqref{eq:wsx12}, we can conclude the desired inclusion $\mathbb{Q}\backslash\{0\}\subseteq S$ as follows. First, note that $-1\in S$ and, because $f_2(-1)=1$, we have that $1\in S$ as well. Let $\alpha\neq -1,0,1$ be an arbitrary  rational number, we will show that $\alpha\in S$. 
\begin{enumerate}
\item  If $\alpha\in (-1,0)$, then $\alpha\in S$ from \eqref{eq:wsx12}. 
\item  If $\alpha\in (0,1)$, let $\alpha':=\alpha-1$ and note that $\alpha'\in (-1,0)\cap \mathbb{Q}$. From \eqref{eq:wsx12}, we have that  $\alpha-1\in S$. Since  $f_1(f_1(f_2(\alpha')))=\alpha$ from \eqref{eq:er123qwe12a}, we conclude that $\alpha\in S$. 
\item  If $\alpha\notin (-1,1)$, then  $-1/\alpha$ is a non-zero rational number in the interval $(-1,1)$ and hence $-1/\alpha\in S$ from the two cases above. It follows that   $f_2(-1/\alpha)=\alpha$ belongs to $S$ as well. 
\end{enumerate}
This finishes the proof of  $\mathbb{Q}\backslash\{0\}\subseteq S$ and therefore establishes Item~\ref{it:SQ} in the beginning of the proof,  yielding Lemma~\ref{lem:gammaminusone}.
\end{proof}

\begin{lemma}\label{lem:gammaminusonequarter}
Let $\gamma=-1/4$. Then, there exists  a tree of maximum degree $\Delta=3$ that implements the edge activity $\gamma'=-1$. 
\end{lemma}
\begin{proof}
As in Lemma~\ref{lem:gammaminusone}, we will omit the argument $\gamma$ from the notation of partition functions (since we will have $\gamma=-1/4$ fixed).

To prove the lemma, we will show that there exists a tree $T$ of maximum degree at most 3 with terminal $u$ that implements the vertex activity $\lambda=4$ and a tree $T'$ of maximum degree at most 3 with terminal $u'$ that implements the vertex activity $\lambda'=1$. From this, we can conclude that there exists a tree that implements the edge activity 
$\gamma'=-1$ using the construction of Lemma~\ref{lem:expmatch}. Namely, let $H_1,H_3$ be disjoint copies of the tree $T$ and $y_1,y_3$ be the corresponding copies of the terminal $u$. Let $H_2$ be the tree $T'$ and let $y_2$ denote the terminal $u'$. Then, for $i\in\{1,2,3\}$, we set 
\begin{equation}\label{eq:b4b4wqw35b}
q_i:=Z_{H_i,\neg y_i},\quad z_i:=Z_{H_i}\neq 0,
\end{equation}
so that $q_1/z_1=q_3/z_3=4$, $q_2/z_2=1$. Let $G$ be the tree obtained by taking the disjoint union of $H_1,H_2,H_3$, two new vertices $u,v$ and adding the edges $(u,y_1), (y_1,y_2), (y_2,y_3),(y_3,v)$. Then, we have that
\begin{equation*}\tag{\ref{eq:ZGuvall}}
\begin{array}{llll}
Z_{G,u,v}=\gamma^2 q_1 z_2 q_3,&  & & Z_{G, \neg u, \neg v}=z_1 z_2 z_3 +\gamma (q_1 q_2 z_3+z_1 q_2 q_3),\\
Z_{G,  u,\neg v}=\gamma q_1 (z_2 z_3+\gamma q_2 q_3),& & & Z_{G, \neg u,v}=\gamma q_3 (z_1 z_2+\gamma q_1 q_2).
\end{array}
\end{equation*}
and hence we obtain that
\[\frac{Z_{G,u,\neg v}}{Z_{G,\neg u,\neg v}}=\frac{Z_{G,\neg u, v}}{Z_{G,\neg u,\neg v}}=0, \quad \frac{Z_{G,u, v}}{Z_{G, \neg u, \neg v}}=-1.\]
Therefore, $G$ with terminals $u,v$ implements the edge activity $\gamma'=-1$, as desired.

It remains to implement the vertex activities $\lambda=4$ and $\lambda'=1$ using trees of maximum degree at most 3. Consider the functions 
\[f_1(x)=\frac{1}{1-\frac{1}{4}x} \mbox{ for $x\neq 4$}, \quad\quad f_2(x)=
\frac{2}{1-\frac{1}{2}x}  \mbox{ for $x\neq 2$}.\]
Analogously to Lemma~\ref{lem:gammaminusone}, we will consider the set $S$ of all numbers $s$ for which there exist an integer $n\geq 0$ and indices $i_1,\hdots,i_n\in\{1,2\}$ such that the sequence given by
\begin{equation}\label{eq:defSbb}
x_0=0, \quad x_{j}=f_{i_j}(x_{j-1}) \mbox{ for }j\in\{1,\hdots,n\}
\end{equation}
satisfies $x_n=s$ (note that  $0\in S$ by taking $n=0$).  We will show that 
\begin{equation}\label{eq:3edcw2}
\begin{array}{l}
\mbox{for indices $i_1,\hdots,i_n\in \{1,2\}$, for any term $\{x_j\}_{\{j=0,1,\hdots,n\}}$ of the sequence in \eqref{eq:defS},}\\
\mbox{there exists a tree $G_j$ of maximum degree at most $3$ and a vertex $u_j$ in $G_j$ such that}\\
\mbox{$\frac{Z_{G_j,\neg u_j}}{Z_{G_j}}=x_j$ and $u_j$ has degree 1 in $G_{j}$ if $j=0$ or $i_j=1$, and degree 2 otherwise.}
\end{array}
\end{equation}
From this, the desired implementations follow by observing that  
\[f_1(f_2(f_1(f_2(f_2(f_1(0))))))=4 ,\quad f_{1}(0)=1.\]

To prove \eqref{eq:3edcw2}, we proceed by induction on $j$.  For the base case $j=0$, consider the binary tree of height 2 and connect its root to a new vertex $u_0$; denote by $G_0$ the tree thus obtained. Then, we have that $Z_{G_0,\neg u_0}=1+6\gamma+8\gamma^2=0$ and  
$Z_{G_0}=Z_{G_0,\neg u_0} +
\gamma {(1+2\gamma)}^2 \neq 0$, so $G_0, u_0$ establish \eqref{eq:3edcw2} in the case $j=0$.  For the inductive step, assume that $j\geq 1$. To obtain $G_{j}$ from $G_{j-1}$, we consider cases on the value of $i_{j}$. 
\begin{enumerate}
\item If $i_{j}=1$, we have $x_{j}=f_1(x_{j-1})$ and therefore $x_{j-1}\neq 4$. Let $G_{j}$ be the graph obtained from $G_j$ by adding a new vertex $u_{j}$ and joining it to $u_{j-1}$. Then, 
\begin{equation*}
Z_{G_{j},\neg u_{j}}=Z_{G_{j-1}} \mbox{ and } Z_{G_{j}}=Z_{G_{j-1}}-\frac{1}{4}Z_{G_{j-1},\neg u_{j-1}}\neq 0,
\end{equation*}
where the disequality follows from $x_{j-1}=\frac{Z_{G_{j-1},\neg u_{j-1}}}{Z_{G_{j-1}}}\neq 4$. We therefore obtain that $\frac{Z_{G_{j},\neg u_{j}}}{Z_{G_{j}}}=\frac{Z_{G_{j-1}}}{Z_{G_{j-1}}-\tfrac14 Z_{G_{j-1},\neg u_{j-1}}}=\frac{1}{1-\frac{1}{4}x_{j-1}}=f_1(x_{j-1})=x_{j}$, as wanted.

\item If $i_{j}=2$, we have $x_{j}=f_2(x_{j-1})$ and therefore $x_{j-1}\neq 2$. 
Let $S$ be a star with three leaves, one of which is called~$\rho$.
Then 
 $Z_{S,\neg \rho}= 1+2\gamma=1/2$ and 
 $Z_{S}=1+3\gamma=1/4$. To obtain  $G_{j}$, we add to $G_{j-1}$ the  star~$S$ and a new vertex $u_{j}$ which is connected to both $u_{j-1}$ and the  leaf $\rho$ of $S$. Then, we have that 
\begin{equation*}
\begin{aligned}
Z_{G_{j},\neg u_{j}}&=Z_{G_{j-1}}Z_{S}=Z_{G_{j-1}}/4,\\
Z_{G_{j}}&=Z_{G_{j-1}}Z_{S}+\gamma(Z_{G_{j-1},\neg u_{j-1}}Z_S+Z_{G_{j-1}}Z_{S,\neg \rho})
=  Z_{G_{j-1}} /8 - Z_{G_{j-1},\neg u_{j-1}}/16
\neq 0,
\end{aligned}
\end{equation*}
where the disequality follows from $x_{j-1}=\frac{Z_{G_{j-1},\neg u_{j-1}}}{Z_{G_{j-1}}}\neq 2$. 
We therefore obtain that 
$$\frac{Z_{G_{j},\neg u_{j}}}{Z_{G_{j}}}=
\frac{
Z_{G_{j-1}}/4
}
{Z_{G_{j-1}}/8 - Z_{G_{j-1},\neg u_{j-1}}/16} =
\frac{2}{1-\tfrac12 \frac{Z_{G_{j-1},\neg u_{j-1}}}{Z_{G_{j-1}}} } = 
\frac{2}{1-\tfrac12  x_{j-1} } =f_2(x_{j-1})=x_{j},$$ as wanted.
\end{enumerate}
This finishes the proof of \eqref{eq:3edcw2}, therefore completing the proof of Lemma~\ref{lem:gammaminusonequarter}.
\end{proof}

\subsection{Concluding the proof of Lemma~\ref{lem:vertmatch}}
In this section, we conclude the proof of Lemma~\ref{lem:vertmatch}. We will need the following technical facts.
\begin{lemma}\label{lem:easydensity}
Let $\gamma\neq 0$ be a real number and suppose that the sequence $\{x_n\}_{n\geq 0}$ is dense in $\mathbb{R}$. Then, for every real $\lambda$ and $\epsilon>0$, there exists $n$ such that $x_n\neq -1/\gamma$ and $|\frac{1}{1+\gamma x_n}-\lambda|\leq \epsilon$. 
\end{lemma}
\begin{proof}
Consider the function $f(x)=1/(1+\gamma x)$ for $x\neq -1/\gamma$. Fix arbitrary real $\lambda$ and $\epsilon>0$, and let $I$ be a closed subinterval of $[\lambda-\epsilon,\lambda+\epsilon]$ that has non-zero length and does not contain the point 0.   Let $g(x)$ be the inverse function of $f$, i.e., $g(x)=(1-x)/(\gamma x)$ for $x\neq 0$, and note that $f(g(x))=x$ for all $x\neq 0$. Since $0\notin I$, $g$ is continuous on the closed interval $I$ and therefore $g(I)$ is also an interval that has non-zero length (from $\gamma\neq 0$). Since the sequence $\{x_n\}_{n\geq 0}$ is dense in $\mathbb{R}$, there exists $n$ such that $x_n\neq -1/\gamma$ and $x_n\in g(I)$. It follows that $f(x_n)\in I\subseteq [\lambda-\epsilon,\lambda+\epsilon]$, as needed. 
\end{proof}

\begin{lemma}\label{lem:impleme12ntation}
Let $\gamma$ be a  real number. Let $G$ be a graph with terminals $u$ and $v$ that implements the edge activity $\gamma'$ (cf. Definition~\ref{def:edgeimpl}). 

For a graph $H=(V_H,E_H)$, let $\widehat{H}$ be the graph obtained by replacing every edge $e$ of $H$ with a distinct copy of the graph $G$ and identifying the endpoints of $e$ with the corresponding copies of the terminals $u$ and $v$. Then, for $C:=(Z_{G,\neg u,\neg v}(\gamma))^{|E_H|}$ and any two vertices $w,z$ of $H$ it holds that $C\neq 0$ and
\begin{equation}\label{eq:rrfr3frv12}
\begin{array}{lll}
Z_{\widehat{H},w}(\gamma)=C\cdot Z_{H,w}(\gamma'),& & Z_{\widehat{H},\neg w}(\gamma)=C\cdot Z_{H, \neg w}(\gamma'),
\end{array}
\end{equation}
and
\begin{equation}\label{eq:rrfr3frv12b}
\begin{array}{lll}
Z_{\widehat{H},w,z}(\gamma)=C\cdot Z_{H,w,z}(\gamma'),& & Z_{\widehat{H},\neg w,\neg z}(\gamma)=C\cdot Z_{H, \neg w, \neg z}(\gamma'),\\
Z_{\widehat{H},w,\neg z}(\gamma)=C\cdot Z_{H,w,\neg z}(\gamma'),& & Z_{\widehat{H},\neg w, z}(\gamma)=C\cdot Z_{H, \neg w,  z}(\gamma').
\end{array}
\end{equation}
\end{lemma}
\begin{proof}
The proof is a simpler version of an analysis appearing in the last part of the proof of Theorems~\ref{thm:hard} and~\ref{thm:hardsign}. Namely, for a matching $\widehat{M}$ of $\widehat{H}$, we will consider the phase of the matching $\widehat{M}$, which we will denote by $\mathcal{Y}(\widehat{M})$. This is a 0-1 vector indexed by pairs $(e,x)$ such that $e$ is an edge of $H$ and $x$ is an endpoint of $e$; we set the $(e,x)$ entry of $\mathcal{Y}(\widehat{M})$ equal to 1 if $x$ is matched in the matching $\widehat{M}$ by edges in the copy of $G$ corresponding to the edge $e$ (and 0 otherwise). 

Let $\mathcal{P}$ be the set  consisting of all phases, i.e., the values of $\mathcal{Y}\big(\widehat{M}\big)$ as $\widehat{M}$ ranges over matchings of $\widehat{H}$.    Fix a possible phase $Y\in \mathcal{P}$ and let $\Omega_Y$ be the set of matchings $\widehat{M}$ of $\widehat{H}$ such that $\mathcal{Y}\big(\widehat{M}\big)=Y$. The aggregate weight of matchings in $\Omega_Y$ is given by 
\begin{equation}\label{eq:ws222A}
W_Y:=\prod_{e=(x,y)\in E_H}t(Y_{(e,x)}, Y_{(e,y)}),
\end{equation}
where
\[\begin{array}{lll} t(0,0):=Z_{G,\neg u,\neg v}(\gamma),& & t(1,1):=Z_{G, u,v}(\gamma),\\
t(0,1):=Z_{G, \neg u, v}(\gamma),& & t(1,0):=Z_{G,u,\neg v}(\gamma).
\end{array}\]
Note that, since the graph $G$ with terminals $u,v$ implements the activity $\gamma'$, we have $t(0,1)=t(1,0)=0$ and $t(1,1)/t(0,0)=\gamma'$. Let $\mathcal{P}_1$ be the set of phases $Y$ in $\mathcal{P}$ such that there exists an edge $e=(x,y)\in E$ such that $Y_{(e,x)}\neq Y_{(e,y)}$ and let $\mathcal{P}_2=\mathcal{P}\backslash P_1$. Then, for all $Y\in \mathcal{P}_1$ we have from \eqref{eq:ws222A} that $W_Y=0$. On the other hand, for $Y\in \mathcal{P}_2$, we have that  the set of edges $M_Y:=\{e\in E_H\mid Y_{(e,x)}= Y_{(e,y)}=1\}$ is a matching in $H$ and $W_Y=(Z_{G,u, v}(\gamma))^{|M|}(Z_{G,\neg u, \neg v}(\gamma))^{|E_H|-|M|}=C\cdot (\gamma')^{|M|}$. 

We can now conclude \eqref{eq:rrfr3frv12} and \eqref{eq:rrfr3frv12b}. For \eqref{eq:rrfr3frv12}, consider an arbitrary vertex $w$ of $H$ and recall that, for a matching $M$, $\mathsf{ver}(M)$ denotes the set of matched vertices in $M$. Then, 
\begin{align*}
Z_{\widehat{H},w}(\gamma)&=\sum_{\widehat{M}\in \mathcal{M}_{\widehat{H}}; w\in \mathsf{ver}(\widehat{M})} \gamma^{|\widehat{M}|}=\sum_{Y\in \mathcal{P}_2; w\in \mathsf{ver}(M_Y) } W_Y=C\sum_{M\in \mathcal{M}_H; w\in \mathsf{ver}(M) }(\gamma')^{|M|}=C\cdot Z_{H,w}(\gamma'),\\
Z_{\widehat{H},\neg w}(\gamma)&=\sum_{\widehat{M}\in \mathcal{M}_{\widehat{H}}; w\notin \mathsf{ver}(\widehat{M})} \gamma^{|\widehat{M}|}=\sum_{Y\in \mathcal{P}_2; w\notin \mathsf{ver}(M_Y) } W_Y=C\sum_{M\in \mathcal{M}_H; w\notin \mathsf{ver}(M) }(\gamma')^{|M|}=C\cdot Z_{H,\neg w}(\gamma').
\end{align*}
The proof of \eqref{eq:rrfr3frv12b} is completely analogous, therefore concluding the proof of the lemma.
\end{proof}
We are now ready to prove Lemma~\ref{lem:vertmatch}, which we restate here for convenience.
\begin{lemvertmatch}
\statelemvertmatch
\end{lemvertmatch}
\begin{proof}
Fix arbitrary $\lambda \in \mathbb{R}$ and $\epsilon>0$. We consider the cases $\gamma\in \mathcal{B}_{\Delta}$ and $\gamma\notin \mathcal{B}_{\Delta}$ separately (where $\mathcal{B}_{\Delta}$ is given in \eqref{def:bdelta}). 

\textbf{Case I: $\gamma\notin \mathcal{B}_{\Delta}$.} For an integer $n\geq 0$, let $T_n$ be the $(\Delta-1)$-ary tree of height $n$ rooted at $\rho_n$. By Item~\ref{it:33rc3d22d} of Lemma~\ref{lem:notbad}, we have that $Z_{T_n}(\gamma)\neq0$ and therefore the sequence $x_n=\frac{Z_{T_n,\neg \rho_n}(\gamma)}{Z_{T_n}(\gamma)}$ is well-defined (for all $n\geq 0$). Moreover, again by Item~\ref{it:33rc3d22d} of Lemma~\ref{lem:notbad}, we have that the sequence $\{x_n\}_{n\geq 0}$ is dense in $\mathbb{R}$ and therefore, from Lemma~\ref{lem:easydensity},  there exists $n$ such that  
\begin{equation}\label{eq:rwew45ce}
x_n\neq -1/\gamma \mbox{\ \ and \ \ } \Big|\frac{1}{1+\gamma x_n}-\lambda\Big|\leq \epsilon.
\end{equation}
Let $T_n'$ be the tree obtained from $T_n$ by adding a new vertex $\rho_n'$ and connecting it to $\rho_n$. Then, 
\[Z_{T_n',\neg \rho_n'}(\gamma)=Z_{T_n}(\gamma) \mbox{\ \ and \ \ } Z_{T_n'}(\gamma)=Z_{T_n}(\gamma)+\gamma Z_{T_n,\neg \rho_n}(\gamma)\neq 0,\]
where the disequality follows from $\frac{Z_{T_n,\neg \rho_n}(\gamma)}{Z_{T_n}(\gamma)}=x_n\neq -1/\gamma$. We also have that $\frac{Z_{T_n',\neg \rho_n'}(\gamma)}{Z_{T_n'}(\gamma)}=\frac{1}{1+\gamma x_n}$. Therefore, from \eqref{eq:rwew45ce}, we obtain that $T_n'$ with terminal $\rho_n'$ implements the vertex activity $\lambda$ with accuracy $\epsilon$.

\textbf{Case II: $\gamma\in \mathcal{B}_{\Delta}$.}  We first show that there exists a tree $T$ of maximum degree at most $\Delta$ with terminals $u,v$ that implements the edge activity $-1$. By Item~\ref{it:33rc3d22d} of Lemma~\ref{lem:notbad}, we  have that there exists a tree $T_1$ of maximum degree at most $\Delta$ with terminals $u_1,v_1$ that implements either the edge activity $-1$ or the edge activity $-1/4$. In the former case, we are clearly done (and we can take $T=T_1$), so consider the latter case. Then, we have that 
\[Z_{T_1,u_1,\neg v_1}(\gamma)=Z_{T_1, \neg u_1,v_1}(\gamma)=0, \quad \frac{Z_{T_1,u_1,v_1}(\gamma)}{Z_{T_1,\neg u_1,\neg v_1}(\gamma)}=-1/4. \] By Lemma~\ref{lem:gammaminusonequarter}, there exists a tree $T_2$ of maximum degree at most 3 with terminals $u_2,v_2$  such that 
\[Z_{T_2,u_2,\neg v_2}(-1/4)=Z_{T_2,\neg u_2,v_2}(-1/4)=0, \quad \frac{Z_{T_2,u_2,v_2}(-1/4)}{Z_{T_2,\neg u_2,\neg v_2}(-1/4)}=-1.\] 
We let $T$ be the tree obtained by replacing every edge $e$ of $T_1$ by a distinct copy of $T_2$ and identifying the endpoints of $e$ with the corresponding copies of the terminals $u_2$ and $v_2$. We let the terminals $u,v$ of $T$ be the terminals $u_1,v_1$ in $T_1$. Note that, since $u_2,v_2$ have degree 1 in $T_2$ and $T_2$ has maximum degree at most 3, the maximum degree of $T$ is at most $\Delta$ and  $u,v$ have degree 1 in $T$. Moreover, Lemma~\ref{lem:impleme12ntation} gives that 
\[Z_{T,u,\neg v}(\gamma)=Z_{T,\neg u,v}(\gamma)=0, \quad \frac{Z_{T,u, v}(\gamma)}{Z_{T,\neg u,\neg v}(\gamma)}=\frac{Z_{T_2,u_2,v_2}(-1/4)}{Z_{T_2,\neg u_2,\neg v_2}(-1/4)}=-1,\] 
and so the tree $T$ implements the edge activity $-1$, as claimed.

Now, to implement the (real) vertex activity $\lambda$ with accuracy $\epsilon>0$, note that by Lemma~\ref{lem:gammaminusone} there exists a tree $H$ of maximum degree at most $3$ with terminal $w$ such that $Z_{H}(-1)\neq 0$ and 
\[\Big|\frac{Z_{H,\neg w}(-1)}{Z_{H}(-1)}-\lambda\Big|\leq \epsilon.\]
($H$ perfectly implements a rational that is close to~$\lambda$.)
We let $G$ be the tree obtained by replacing every edge $e$ of $H$ by a distinct copy of $T$ and identifying the endpoints of $e$ with the corresponding copies of the terminals $u$ and $v$ in $T$. We let the terminal $u$ of $G$ be the terminal $w$ of $H$. As before, we have that $G$ has maximum degree at most $\Delta$ and that $u$ has degree 1 in $G$. Further, Lemma~\ref{lem:impleme12ntation} gives that $Z_{G}(\gamma)\neq 0$ and  
\[\frac{Z_{G,\neg u}(\gamma)}{Z_{G}(\gamma)}=\frac{Z_{H,\neg w}(-1)}{Z_{H}(-1)},\] 
and so the tree $G$ implements the vertex activity $\lambda$ with accuracy $\epsilon$,  as needed.

This finishes the proof of Lemma~\ref{lem:vertmatch}.
\end{proof}

\subsection{Implementing activities with exponential precision}
In this section, we bootstrap Lemma~\ref{lem:vertmatch} to obtain the exponential precision required in Lemma~\ref{lem:expvertmatch}. The following lemma lies at the heart of the argument and is based on the ``contracting maps that cover'' technique of \cite{complexhard}. 
\begin{lemma}\label{lem:exp2}
Let $\gamma<0$ be a rational number. Then, there exist rationals $x_0$ and $r,\delta>0$ and reals $\lambda_1^*,\hdots,\lambda_t^*$ (for some positive integer $t$) such that the following holds for all rational $\lambda_1,\hdots, \lambda_t$ satisfying $|\lambda_i-\lambda_i^*|\leq \delta$ for $i\in [t]$.

Let $I:=[x_0-r,x_0+r]$ and, for $i\in [t]$, consider the  map $\Phi_i:x\mapsto \frac{1}{1+\gamma(\lambda_i+ x)}$ for $x\neq -(1+\gamma \lambda_i)/\gamma$. There is an algorithm which, on input (i) a starting point $y_0\in I\cap \mathbb{Q}$, (ii) a target $y\in I\cap \mathbb{Q}$,  and (iii) a rational $\epsilon>0$, outputs in $poly(\size{y_0,y,\eps})$ time a number $\hat{y}\in I\cap \mathbb{Q}$ and  a sequence $i_1,i_2,\hdots,i_k\in [t]$ such that  
\[\hat{y}=\Phi_{i_k}(\Phi_{i_{k-1}}(\cdots\Phi_{i_1}(y_0)\cdots))\mbox{ and }|\hat{y}-y|\leq \epsilon.\]
\end{lemma}
\begin{proof}
Let $x_1,x_2$ be rationals such that $\gamma x_1 x_2=-1$ and $x_1\neq \pm x_2$. Let $\lambda$ be such that $1+\gamma \lambda=-\gamma(x_1+x_2)$. Then,   the fixpoints of the map $\Phi: x\mapsto \frac{1}{1+\gamma(\lambda+ x)}$ are $x_1$ and $x_2$, and at least one of the two points is attracting.\footnote{To see this, note that $\Phi(x)=x$ is equivalent to  $x(1+\gamma \lambda)+\gamma x^2=1$ and therefore $x_1$ and $x_2$ are (the only) fixpoints of $\Phi$. Moreover, we have that $\Phi'(x)=-\frac{\gamma}{(1+\gamma(\lambda+x))^2}$ and hence $\Phi'(x_1)=-\gamma x_1^2$, $\Phi'(x_2)=-\gamma x_2^2$. Therefore $|\Phi'(x_1)|\neq |\Phi'(x_2)|$ and   $1=|\gamma x_1 x_2|=\sqrt{|\Phi'(x_1)| |\Phi'(x_2)|}$. Therefore either $|\Phi'(x_1)|<1$ or $|\Phi'(x_2)|<1$.} Denote by $x_0$ the attracting fixpoint of $\Phi$, so that $x_0$ satisfies $\Phi(x_0)=x_0$ and $0<|\Phi'(x_0)|<1$.  By Lemma~\ref{lem:Ratioapprox}, there exists $\eta>0$ such that for all $x\in[x_0-\eta,x_0+\eta]$ and all $\lambda'\in [\lambda-\eta,\lambda+\eta]$ it holds that  
\begin{equation}\label{eq:jntbr45}
\mbox{$1+\gamma(\lambda'+x)\neq 0$ and }\bigg|\frac{\gamma}{\big(1+\gamma(\lambda'+x)\big)^2}-\frac{\gamma}{\big(1+\gamma(\lambda+x_0)\big)^2}\bigg|\leq \frac{1}{2}\min\big\{|\Phi'(x_0)|,1-|\Phi'(x_0)|\big\}.
\end{equation}
Let $r:=\frac{|\Phi'(x_0)|}{4}\eta$, $\delta:=(r/4)$ and  let $\lambda_1^*,\hdots, \lambda_t^*$ form a $\delta$-covering of the interval $[\lambda-\eta/2,\lambda+\eta/2]$. Let $\lambda_1,\hdots, \lambda_t$ be arbitrary rationals satisfying $|\lambda_i-\lambda_i^*|\leq \delta$. For $i\in [t]$ consider the maps $\Phi_i:x\mapsto \frac{1}{1+\gamma(\lambda_i+x)}$. Finally, let $I$ be the interval $[x_0-r,x_0+r]$.
We will show 
\begin{description}
\item[Property 1:]
The  maps $\{\Phi_i\}_{i\in [t]}$ are contracting on the interval $I$, and
\item[Property 2:] $I\subseteq \Phi_1(I)\cup \cdots \cup \Phi_t(I)$. 
\end{description}
Once these two properties of the maps $\{\Phi_i\}_{i\in [t]}$ are proved, the algorithm in the statement of the lemma and its analysis are almost identical to those in \cite[Proof of Lemmas 12 \& 26]{complexhard}. The only difference here is that the maps $\{\Phi_i\}_{i\in [t]}$ have  different expressions.
The fact that we need about the expression of the maps is that, for $i\in [t]$ and 
for every rational $x$, $\Phi^{-1}_i(x)$ can be computed in time $poly(\size{x,\lambda_i,\gamma})$. This is clear since $\Phi^{-1}_i(x)=\frac{1}{\gamma}(\frac{1}{x}-1)-\lambda_i$.

{\bf Proof of Property 1:\quad} 
Fix $i\in[t]$. We will show that $\Phi_i$ is contracting on the interval $I$. Observe that  $r<\eta/4$ since $|\Phi'(x_0)|<1$ and therefore $\delta<\eta/4$ as well.  Then, we have by the triangle inequality that 
\[|\lambda_i-\lambda|\leq |\lambda_i-\lambda_i^*|+|\lambda_i^*-\lambda|\leq \delta+\eta/2<\eta.\]
Therefore, we can apply \eqref{eq:jntbr45} to $\lambda'=\lambda_i$ and $x\in I$. Observe that $\Phi'(x)=-\gamma/\big(1+\gamma(\lambda_i+x)\big)^2$ and $\Phi'(x_0)=-\gamma/\big(1+\gamma(\lambda_i+x_0)\big)^2$ 
and hence we obtain that for all $x\in I$ it holds that 
\[|\Phi_i'(x)|\leq \frac{1}{2}(1+|\Phi'(x_0)|)<1.\]
It follows that the maps $\Phi_i$ are contracting on the interval $I$ for all $i\in [t]$. 

{\bf Proof of Property 2:\quad} It suffices to consider an arbitrary $y\in I$ and show that there exists $j\in [t]$ such that $\Phi_j^{-1}(y)\in I$. To do this, we set $J$ to be the interval $[x_0-\eta/2,x_0+\eta/2]$ and consider the map $\Phi$ on the interval $J$. Then, \eqref{eq:jntbr45} for $\lambda'=\lambda$ and $x\in J$  gives that  
\[0<\frac{1}{2}|\Phi'(x_0)|\leq |\Phi'(x)|,\]
and therefore, by the Mean Value Theorem, for $z,w\in J$ we have that 
\begin{equation}\label{eq:phixO352}
\frac{1}{2}|\Phi'(x_0)|\cdot |z-w|\leq |\Phi(z)-\Phi(w)|.
\end{equation}
We thus have that 
\begin{equation*}
\begin{aligned}
|\Phi(x_0+\eta/2)-x_0|&=|\Phi(x_0+\eta/2)-\Phi(x_0)|\geq \eta |\Phi'(x_0)|/4=r,\\
|\Phi(x_0-\eta/2)-x_0|&=|\Phi(x_0-\eta/2)-\Phi(x_0)|\geq \eta |\Phi'(x_0)|/4=r.
\end{aligned}
\end{equation*}
Since $\Phi$ is monotonically increasing and continuous on the interval $J$, we therefore obtain that $I\subseteq \Phi(J)$. Therefore, for arbitrary $y\in I$ it holds that $\Phi^{-1}(y)\in J$ and hence from \eqref{eq:phixO352} applied to $z=\Phi^{-1}(y)$ and $w=\Phi^{-1}(x_0)$,   we obtain that 
\[|\Phi^{-1}(y)-x_0|=|\Phi^{-1}(y)-\Phi^{-1}(x_0)|\leq (2/|\Phi'(x_0)|) (y-x_0)\leq \eta/2.\]
Since $\lambda_1^*,\hdots, \lambda_t^*$ is a $\delta$-covering of the interval $[\lambda-\eta/2,\lambda+\eta/2]$, it follows that there exists $j\in [t]$ such that 
\[\big|\lambda+\Phi^{-1}(y)-x_0-\lambda_j^*\big|\leq \delta= r/4.\] 
Now, observe that $\Phi^{-1}_j(y)=\frac{1}{\gamma}\big(\frac{1}{y}-1\big)-\lambda_j$ and $\Phi^{-1}(y)=\frac{1}{\gamma}\big(\frac{1}{y}-1\big)-\lambda$, so we have that 
\begin{align*}
\big|\Phi^{-1}_j(y)-x_0\big|&=\bigg|\frac{1}{\gamma}\Big(\frac{1}{y}-1\Big)-\lambda_j-x_0\bigg|=|\lambda+\Phi^{-1}(y)-x_0-\lambda_j|\\
&\leq |\lambda+\Phi^{-1}(y)-x_0-\lambda_j^*|+|\lambda_j-\lambda_j^*|\leq r/4+r/4=r/2.
\end{align*}
It follows that $y\in \Phi_j(I)$ and therefore, since $y$ was arbitrary, we have that $I\subseteq \Phi_1(I)\cup \cdots \cup \Phi_t(I)$.

This completes the proof of Properties~1 and~2, and hence the proof of Lemma~\ref{lem:exp2}.
\end{proof}

We are now ready to give the proof of Lemma~\ref{lem:expvertmatch}, which we restate here for convenience.
\begin{lemexpvertmatch}
\statelemexpvertmatch
\end{lemexpvertmatch}
\begin{proof}[Proof of Lemma~\ref{lem:expvertmatch}]
Let $x_0,\delta, r$ and $\lambda_1^*, \lambda_2^*,\hdots,\lambda^*_t$ be the constants from Lemma~\ref{lem:exp2} ($x_0$, $r$ and $\delta$ are rationals with $r,\delta>0$ and 
$\lambda_1^*, \lambda_2^*,\hdots,\lambda^*_t$ are reals, for some positive integer~$t$.
These depend on~$\gamma$ but they do not depend on~$\lambda$ and~$\epsilon$,
which are the inputs to the algorithm described in this lemma.)
Set $\delta'=\min\{\delta,r/2\}$. Let also $\lambda^*_0$ be the rational number given by 
\begin{equation}\label{eq:lambda0x0}
1+\gamma+\gamma \lambda_0^*+\gamma x_0=0.
\end{equation}
By Lemma~\ref{lem:vertmatch}, for $i=0,1,\hdots,t$, there exists a bipartite graph $G_i$ of maximum degree $\Delta$ which implements the vertex activity $\lambda_i^*$ with accuracy $\delta'$
(once again, the size and time to construct~$G_i$ is independent of~$\lambda$ and~$\epsilon$). Let $v_i$ be the terminal of $G_i$ and set $\lambda_i=\frac{Z_{G_i,\neg v_i}(\gamma)}{Z_{G_i}(\gamma)}$, so that  $\lambda_i\in [\lambda_i^*-\delta',\lambda_i^*+\delta']$.  Similarly, let $H_0$ with terminal $w_0$ be a bipartite graph  of maximum degree $\Delta$ which implements the vertex activity $x_0$ with accuracy $\delta'$, and set $y_0=\frac{Z_{H_0,\neg w_0}(\gamma)}{Z_{H_0}(\gamma)}$ so that $y_0\in [x_0-\delta',x_0+\delta']$.

\vskip 0.15cm

Now, suppose that we are given inputs $\lambda,\eps\in \mathbb{Q}$ with $\eps>0$  and we want to output in $poly(\size{\lambda,\eps})$ time a bipartite graph of maximum degree $\Delta$ that implements the vertex activity $\lambda$ with accuracy $\epsilon$. We will consider two cases for the range of $\lambda$.
\vskip 0.2cm
\noindent \textbf{Case I ($\lambda$ away from 1):} $|\lambda-1|\geq 2/r$.  Let $y$ be the rational given by
\begin{equation}\label{eq:3fgg56yy35tg}
y=-\Big(\frac{1}{\gamma}+\lambda_0+1+\frac{1}{\lambda-1}\Big), \mbox{ so that }\frac{1+\gamma(\lambda_0+y)}{1+\gamma(1+\lambda_0+ y)}=\lambda.
\end{equation}
From \eqref{eq:lambda0x0} we have that $x_0=-(\frac{1}{\gamma}+\lambda_0^*+1)$, so using the triangle inequality we have that
\[|y-x_0|\leq |\lambda_0-\lambda_0^*|+\frac{1}{|\lambda-1|}\leq r/2+r/2=r.\]
Therefore, $y$ belongs to the interval  $I:=[x_0-r,x_0+r]$. Note also that $1+\gamma(1+\lambda_0+ y)=\frac{\gamma}{\lambda-1}\neq 0$.

By Lemma~\ref{lem:Ratioapprox}, there exists $\epsilon'$ with size $poly(\size{y,\epsilon})=poly(\size{\lambda,\epsilon})$ such that for all $x$ satisfying $|x-y|\leq \epsilon'$ it holds that
\begin{equation}\label{eq:3v3r485vr12ew}
1+\gamma(1+\lambda_0+x)\neq 0 \quad\mbox{and}\quad\Big|\frac{1+\gamma(\lambda_0+x)}{1+\gamma(1+\lambda_0+ x)}-\frac{1+\gamma(\lambda_0+y)}{1+\gamma(1+\lambda_0+ y)}\Big|\leq \epsilon.
\end{equation}
For $i=0,1,\hdots, t$, let  $\Phi_i$ be the map $x\mapsto\frac{1}{1+\gamma(\lambda_i+x)}$ for $x\neq -(1+\gamma \lambda_i)/\gamma$. Then, since each $\lambda_i$ is a rational in the interval $[\lambda_i^*-\delta,\lambda_i^*+\delta]$ and $y,y_0$ are rationals in the interval $[x_0-r,x_0+r]$, using the algorithm of  Lemma~\ref{lem:exp2} we  obtain in time $poly(\size{y_0,y,\eps'})=poly(\size{\lambda,\eps})$ a number $\hat{y}$ and a sequence $i_1,\hdots,i_k\in\{1,\hdots,t\}$ such that 
\begin{equation}\label{eq:vv879rf2}
\hat{y}=\Phi_{i_k}(\Phi_{i_{k-1}}(\cdots\Phi_{i_1}(y_0)\cdots))\mbox{ and }|\hat{y}-y|\leq \epsilon'.
\end{equation}
Recall that $H_0$ with terminal $w_0$ implements the vertex activity $y_0$. For $j=1,\hdots, k$, we will define a graph $H_j$ and a vertex $w_j$ in $H_j$ which has degree 2 as follows. Take the graph $H_{j-1}$ and a (new) copy of the graph $G_{i_j}$ (that implements the vertex activity $\lambda_{i_j}$) and add a new vertex $w_j$ whose neighbours are the 
vertex $w_{j-1}$ of $H_{j-1}$ and the terminal $v_{i_j}$ of $G_{i_j}$. Since $H_0,G_1,\hdots,G_{t}$ are bipartite graphs with maximum degree $\Delta$, we have that $H_1,\hdots, H_k$ are bipartite graphs with maximum degree $\Delta$ as well (using that $\Delta\geq 3$ and that $w_0,v_1,\hdots,v_t$ have degree 1 in $H_0,G_1,\hdots, G_t$, respectively). We will show by induction that for all $j=0,1,\hdots, k$ it holds that
\begin{equation}\label{eq:4nrtbrbr4v}
Z_{H_j}(\gamma)\neq 0 \quad\mbox{and}\quad \frac{Z_{H_j,\neg w_j}(\gamma)}{Z_{H_j}(\gamma)}=y_j \mbox{ where, for $j\geq 1$, } y_j:=\Phi_{i_j}(\Phi_{i_{j-1}}(\cdots\Phi_{i_1}(y_0)\cdots)).
\end{equation}
Since $H_0$ with terminal $w_0$ implements $y_0$, \eqref{eq:4nrtbrbr4v} is immediate for $j=0$. Suppose that \eqref{eq:4nrtbrbr4v} holds for some $j\in\{0,\hdots,k-1\}$. Observe  that 
\begin{equation*}
\begin{aligned}
Z_{H_{j+1},\neg w_{j+1}}(\gamma)&=Z_{H_j}(\gamma)Z_{G_{i_j}}(\gamma),\\
Z_{H_{j+1}}(\gamma)&=Z_{H_j}(\gamma)Z_{G_{i_j}}(\gamma)+\gamma Z_{H_j,\neg w_j}(\gamma)Z_{G_{i_j}}(\gamma)+\gamma Z_{H_j}(\gamma)Z_{G_{i_j},\neg v_{i_j}}(\gamma).
\end{aligned}
\end{equation*}
Moreover, we have that $y_{j+1}=\Phi_{i_j}(y_j)$ which implies (using the definition of~$\Phi_{i_j}$) that $y_{j}\neq -(1+\gamma \lambda_{i_{j}})/\gamma$. It follows that 
\[Z_{H_{j+1}}(\gamma)=Z_{H_j}(\gamma)Z_{G_{i_j}}(\gamma)\big(1+(\gamma\lambda_{i_{j+1}}+y_j)\big)\neq 0,\]
and therefore 
$\frac{Z_{H_{j+1},\neg w_{j+1}}(\gamma)}{Z_{H_{j+1}}(\gamma)}= 
\frac{1}{1+\gamma(\lambda_{i_{j+1}}+y_j)}=\Phi_{i_{j+1}}(y_j)=y_{j+1}$. 
This  completes the proof of \eqref{eq:4nrtbrbr4v}. 
From \eqref{eq:4nrtbrbr4v} with $j=k$ and  \eqref{eq:vv879rf2}, we obtain that 
\begin{equation}\label{eq:4nr64hrg666tbrbr4v}
Z_{H_k}(\gamma)\neq 0 \quad\mbox{and}\quad \frac{Z_{H_k,\neg w_k}(\gamma)}{Z_{H_k}(\gamma)}=\hat{y}.
\end{equation}
Recall that $G_0$ with terminal $v_0$ implements the vertex activity $\lambda_0$. Let $G$ be the bipartite graph obtained by taking the disjoint union of $H_k$ and $G_0$,  connecting $w_k, v_0$ to a new vertex $z$, and connecting $z$ to a new vertex $v$. Then, we have that
\begin{equation}\label{eq:bbbbrbu5433357h5}
\begin{aligned}
Z_{G,\neg v}(\gamma)&=Z_{H_k}(\gamma)Z_{G_0}(\gamma)+\gamma Z_{H_k,\neg w_k}(\gamma)Z_{G_0}(\gamma)+\gamma Z_{H_k}(\gamma)Z_{G_0,\neg v_0}(\gamma),\\
Z_{G}(\gamma)&=Z_{G,\neg v}(\gamma)+Z_{G,v}(\gamma)=Z_{G,\neg v}(\gamma)+\gamma Z_{H_k}(\gamma)Z_{G_0}(\gamma).
\end{aligned}
\end{equation}
Using \eqref{eq:4nr64hrg666tbrbr4v} and  \eqref{eq:bbbbrbu5433357h5}, it follows that 
\[Z_G(\gamma)=Z_{H_k}(\gamma)Z_{G_0}(\gamma)\big(1+\gamma(1+\lambda_0+\hat{y})\big)\neq 0,\]
where the disequality follows from $Z_{G_0}(\gamma)\neq 0$ and the disequalities in \eqref{eq:3v3r485vr12ew} and \eqref{eq:4nr64hrg666tbrbr4v} (the former is applied for $x=\hat{y}$; \eqref{eq:vv879rf2} guarantees that \eqref{eq:3v3r485vr12ew} indeed applies). We thus obtain that $G$ with terminal $v$ implements the vertex activity $\frac{Z_{G,\neg v}(\gamma)}{Z_{G}(\gamma)}=\frac{1+\gamma(\lambda_0+\hat{y})}{1+\gamma(1+\lambda_0+\hat{y})}$ which from \eqref{eq:3v3r485vr12ew} is within distance $\epsilon$ from $\lambda$, as required.

\vskip 0.2cm

\noindent \textbf{Case II ($\mathbf{\lambda}$ close to 1):} $|\lambda-1|<2/r$. We first consider the case where $\lambda\neq 1$. We can compute rational numbers $y_1, y_2$ in time $poly(\size{\lambda})$ such that $|y_1-1|,|y_2-1|\geq 2/r$ and 
\begin{equation}\label{eq:409f43rffrfr}
y_1+y_2=-\frac{1}{\gamma}-\frac{\lambda}{\lambda-1}, 
\mbox{ so that } \frac{1+\gamma(y_1+y_2)}{1+\gamma(1+y_1+y_2)}=\lambda.
\end{equation}
By Lemma~\ref{lem:Ratioapprox}, there exists $\epsilon'$ with size $poly(\size{y_1,y_2,\epsilon})=poly(\size{\lambda,\epsilon})$ such that for all $x_1,x_2$ satisfying $|x_1-y_1|,|x_2-y_2|\leq \epsilon'$ it holds that
\begin{equation}\label{eq:3v3rvr12e5464w}
1+\gamma(1+x_1+x_2)\neq 0 \quad\mbox{and}\quad\Big|\frac{1+\gamma(x_1+x_2)}{1+\gamma(1+x_1+x_2)}-\frac{1+\gamma(y_1+y_2)}{1+\gamma(1+y_1+y_2)}\Big|\leq \epsilon.
\end{equation}
By Case I, for $i=1,2$ we can construct a bipartite graph $J_i$ of maximum degree $\Delta$ with terminal $u_i$ that  implements the vertex activity $y_i$ with accuracy $\epsilon'$ in time $poly(\size{y_i,\epsilon'})=poly(\size{\lambda,\epsilon})$.  Let $\hat{y}_i= \frac{Z_{J_i,\neg u_i}(\gamma)}{Z_{J_i}}$, so that $|\hat{y}_i-y_i|\leq \epsilon'$. Therefore, from \eqref{eq:409f43rffrfr} and \eqref{eq:3v3rvr12e5464w}, we obtain that 
\begin{equation}\label{eq:3v3rvr12ew}
1+\gamma(1+\hat{y}_1+\hat{y}_2)\neq 0\quad\mbox{and}\quad\Big|\frac{1+\gamma(\hat{y}_1+\hat{y}_2)}{1+\gamma(1+\hat{y}_1+\hat{y}_2)}-\lambda\Big|\leq \epsilon.
\end{equation}
The rest of the argument is completely analogous to the last part of the argument in Case I. Namely,  let $G$ be the bipartite graph of maximum degree $\Delta$ obtained by taking the disjoint union of  $J_1$ and $J_2$, connecting their terminals $u_1,u_2$ to a new vertex $z$, and connecting $z$ to a new vertex $u$. Then, we have that
\begin{equation*}
\begin{aligned}
Z_{G,\neg u}(\gamma)&=Z_{J_1}(\gamma)Z_{J_2}(\gamma)+\gamma Z_{J_1,\neg u_1}(\gamma)Z_{J_2}(\gamma)+\gamma Z_{J_1}(\gamma)Z_{J_2,\neg u_2}(\gamma),\\
Z_{G}(\gamma)&=Z_{G,\neg u}(\gamma)+Z_{G,u}(\gamma)=Z_{G,\neg u}(\gamma)+\gamma Z_{J_1}(\gamma)Z_{J_2}(\gamma).
\end{aligned}
\end{equation*}
It follows that 
\[Z_G(\gamma)=Z_{J_1}(\gamma)Z_{J_2}(\gamma)\big(1+\gamma(1+\hat{y}_1+\hat{y}_2)\big)\neq 0,\]
where the disequality follows from $Z_{J_1}(\gamma)Z_{J_2}(\gamma)\neq 0$ (since $J_1,J_2$ implement $\hat{y}_1,\hat{y}_2$) and the disequality in \eqref{eq:3v3rvr12ew}. We thus obtain that $G$ with terminal $u$ implements the vertex activity $\frac{Z_{G,\neg u}(\gamma)}{Z_{G}(\gamma)}=\frac{1+\gamma(\hat{y}_1+\hat{y}_2)}{1+\gamma(1+\hat{y}_1+\hat{y}_2)}$ which from \eqref{eq:3v3rvr12ew} is within distance $\epsilon$ from $\lambda$, as required.

To finish Case II, it remains to consider the case where $\lambda=1$. Then we set $\epsilon'=\min\{2/r,\epsilon\}$ and use the preceding method to  implement the activity $1+\frac{1}{2}\epsilon'\neq 1$ with accuracy $\epsilon'/2$ in time $poly(\size{\epsilon'})=poly(\size{\epsilon})$. The  implemented vertex activity $\hat{\lambda}$ satisfies by the triangle inequality $|\hat{\lambda}-1|\leq \epsilon'\leq \epsilon$, as needed.

This completes the description and analysis of the algorithm in all cases, thus completing the proof of Lemma~\ref{lem:expvertmatch}.
\end{proof}

\section{Hardness for bounded connective constant}
In this section, we prove Theorem~\ref{thm:hardconnconst}.
The basic idea is to take an instance~$G$ with a parameter~$\gamma^*$
from a bounded-degree regime 
where hardness is already established 
and add gadgets that don't change the connective constant
so that the partition function of~$G$ may be deduced from the partition function
of a new instance~$G'$ with small connective constant, using the desired value of~$\gamma$.

\begin{proof}[Proof of Theorem~\ref{thm:hardconnconst}]
Recall that for an integer $d\geq 3$, the set $\mathcal{B}_d$ is given by
\begin{equation*}\tag{\ref{def:bdelta}}
\mathcal{B}_{d}=\Big\{\gamma\in \mathbb{R}\mid \gamma=-\frac{1}{4(d-1)(\cos \theta)^2} \mbox{ for some $\theta\in (0,\pi/2) $ which is a rational multiple of $\pi$}\Big\}.
\end{equation*}
Let $S=\bigcup_{d\geq 3} \mathcal{B}_d$, then we have that $S$ is dense on the negative real axis. Henceforth, fix $\gamma\in S$ and fix arbitrary real numbers $\Delta>1$ and $a,c>0$. We will next establish that Items~\ref{it:14g442223} and~\ref{it:14g442223b} in the statement of the theorem hold. 

Let $\gamma^*=-1$.  For each integer $k\geq 0$, we will display a tree $T_k$ with terminals $u_k,v_k$ that implements the edge activity $\gamma^*=-1$, whose terminals $u_k$ and $v_k$ have distance at least $k$ in $T_k$.
Assuming this for the moment, we can conclude the \#P-hardness results of Items~\ref{it:14g442223} and~\ref{it:14g442223b} as follows. 

By Theorems~\ref{thm:hard} and~\ref{thm:hardsign}, $\#\mathsf{BipMatchings}(\gamma^*,3,1.01)$ and $\#\mathsf{SignMatchings}(\gamma^*,3)$ are \#P-hard. Set $k'=\left\lceil 10\log_2 \Delta\right\rceil$ and set for convenience $T'=T_{k'}$, $u'=u_{k'}$ and $v'=v_{k'}$. Denote also by $n'$ the number of vertices in $T'$.    Let $G=(V,E)$ be a graph of maximum degree 3 and suppose that the number of vertices $|V|$ is sufficiently large such that $24(n'!)^{2}\leq c\Delta^{a\log_2((n'-2)|V|)}$. Let $G'=(V',E')$ be the graph obtained from $G$ by replacing each edge $(z,w)$ of $G$ by a new copy of the tree  $T'$ and identifying $z$ and $w$ with the terminals $u'$ and $v'$ of $T'$. Note that $|V'|\geq (n'-2)|V|$. Moreover, by Lemma~\ref{lem:impleme12ntation}, we have that
\[Z_{G'}(\gamma)=(Z_{T',\neg u',\neg v'}(\gamma))^{|E|}\cdot Z_{G}(\gamma^*).\]
Using the \#P-hardness results of Theorems~\ref{thm:hard} and~\ref{thm:hardsign}, we therefore obtain that, in order to conclude Items~\ref{it:14g442223} and~\ref{it:14g442223b}, it  only suffices to show that $G'$ belongs to the family $\mathcal{F}_{\Delta,a,c}$ of graphs with connective constant at most $\Delta$ and profile $(a,c)$.  To do this, we will map a path $P'$ of $G'$ to a path $P$ of $G$ by just considering the vertices of $G$ that the path $P'$ traverses (recall from the construction above that the terminals of the copies of the tree $T'$ are identified to vertices of $G$).  Since the terminals $u'$ and $v'$ of $T'$ are at distance $\geq k'$ in $T'$, for a path $P'$ of $G$ with $\ell$ edges we have that its image $P$ under the map has at most $\ell/k'$ edges. Conversely, every path $P$ in $G$ is the image of at most $4(n'!)$ paths $P'$ in $G'$ under the map (since $T'$ is a tree, for every edge of $P$, the path $P'$ must traverse the unique path connecting the terminals $u'$ and $v'$ of the corresponding copy of $T'$; the only freedom is for the choice of the starting and ending subpaths of $P'$ which must be paths of $T'$ and there are crudely at most $4(n'!)^{2}$ such paths). For any vertex $v$ of $G$ there are at most $3\cdot 2^{\ell-1}$ paths with $\ell$ edges starting from $v$, so for any vertex $w$ of $G'$ we have that 
\[\sum^\ell_{i=1} N_{G'}(w,i)\leq 12(n'!)^{2}\sum^\ell_{i=1} 2^{\left\lceil i/k'\right\rceil-1}\leq 24(n'!)^{2}2^{\ell/k'}\leq c \Delta^\ell \mbox{ for all } \ell\geq a\log |V'|.\]
It follows that $G'\in \mathcal{F}_{\Delta,a,c}$, completing the reduction assuming the existence of the sequence of trees $T_k$ implementing the edge activity $\gamma^*=-1$.

To construct the sequence of trees $T_k$, first note that, since $\gamma \in S$, we have that $\gamma\in \mathcal{B}_d$ for some $d\geq 3$. In the proof of Lemma~\ref{lem:vertmatch}, we showed that for any such $\gamma$ there exists a tree $T$ of maximum degree $d$ with terminals $u$ and $v$ that implements the edge activity $\gamma^*=-1$ (perfectly), cf. the relevant Lemmas~\ref{lem:notbad} and~\ref{lem:gammaminusonequarter}. We will now bootstrap this construction as follows. For $k\geq 1$, let $P_k$ be a path with $k$ vertices (and $k-1$ edges) and let $u_k,v_k$ be the endpoints of the path. Obtain the tree $T_k'$ from $P_k$ by replacing every edge $(z,w)$ of the path $P_k$ with a new copy of the tree $T$ and identifying the endpoints $z$ and $w$ of the edge with the terminals $u$ and $v$ of $T$. We will show that for any integer $k\geq 0$, the tree $T_{6k+5}'$ implements the edge activity $\gamma^*=-1$; the desired sequence $T_k$ is obtained by taking $T_{k}=T_{6k+5}'$. 

To prove that the tree $T_{6k+5}'$ implements the edge activity $\gamma^*=-1$, set for convenience $n=6k+5$ and note that, since $T$ implements the edge activity $\gamma^*=-1$, it suffices by Lemma~\ref{lem:impleme12ntation} to show that
\begin{equation}\label{eq:rtnrtn44}
\frac{Z_{P_{n},u_{n},v_{n}}(\gamma^*)}{Z_{P_{n},\neg u_{n},\neg v_{n}}(\gamma^*)}=-1, \quad Z_{P_{n},\neg u_{n},v_{n}}(\gamma^*)=Z_{P_{n},u_{n},\neg v_{n}}(\gamma^*)=0.
\end{equation}
Using $Z_{P_{i}}(\gamma^*)=Z_{P_{i-1}}(\gamma^*)+\gamma^*Z_{P_{i-2}}(\gamma^*)$ for integer $i\geq 2$, we obtain by induction on $k$ that 
\[Z_{P_{6k+1}}(\gamma^*)=1,  Z_{P_{6k+2}}(\gamma^*)=0, Z_{P_{6k+3}}(\gamma^*)=-1, Z_{P_{6k+4}}(\gamma^*)=-1, Z_{P_{6k+5}}(\gamma^*)=0, Z_{P_{6k+6}}(\gamma^*)=1,\]
and therefore
\begin{equation*}
\begin{array}{lcl}
Z_{P_{n},u_{n},v_{n}}(\gamma^*)=(\gamma^*)^2 Z_{P_{n-4}}(\gamma^*)=1,& & Z_{P_{n},\neg u_{n},\neg v_{n}}(\gamma^*)= Z_{P_{n-2}}(\gamma^*)=-1,\\
Z_{P_{n},\neg u_{n},v_{n}}(\gamma^*)=\gamma^*\cdot Z_{P_{n-3}}(\gamma^*)=0,& & Z_{P_{n},u_{n},\neg v_{n}}(\gamma^*)=\gamma^*\cdot Z_{P_{n-3}}(\gamma^*)=0.
\end{array}
\end{equation*}
This proves \eqref{eq:rtnrtn44} and hence, for each integer $k\geq0$, the tree $T_{6k+5}'$ with terminals $u_{3k+5},v_{3k+5}$ implements the edge activity -1, completing the proof.
\end{proof}

\bibliographystyle{plain}
\bibliography{\jobname}

\end{document}